\documentclass[11pt]{article}
\usepackage[papersize={8.5in,11in},top=1in,left=1in,right=1in,bottom=1in]{geometry}
\usepackage{amsmath}
\usepackage{amssymb}
\usepackage{amsthm}
\usepackage{listings}
\usepackage{cite}

\allowdisplaybreaks[1]
\theoremstyle{definition}
\newtheorem{thm}{Theorem}[section]
\newtheorem{mydef}[thm]{Definition}

\newtheorem{cor}[thm]{Corollary}
\newtheorem{lem}[thm]{Lemma}
\newtheorem{claim}[thm]{Claim}

\renewcommand{\vec}[1]{\mathbf{#1}}
\newcommand{\sseg}{Semi-Smooth Games}
\newcommand{\sse}{semi-smooth}
\newcommand{\sseness}{semi-smoothness}
\newcommand{\maxcut}{Sharing with Conflicts}
\newcommand{\minuncut}{Balancing with Conflicts}
\newcommand{\maxuncut}{Sharing with Friendship}
\newcommand{\mincut}{Balancing with Friendship}
\newcommand{\mincorrelation}{Balancing with Conflicts and Friendship}

\begin{document}

\title{Assignment Games with Conflicts:\\ Price of Total Anarchy and Convergence Results via Semi-Smoothness\thanks{This work was supported in part by NSF grants CCF-0914782, CCF-1101495, and CNS-1017932.}}
\author{Elliot Anshelevich\\ Rensselaer Polytechnic Institute, Troy, NY 12180 \and John Postl \\ Rensselaer Polytechnic Institute, Troy, NY 12180 \and Tom Wexler \\ Oberlin College, Oberlin, OH 44074}
\maketitle

\begin{abstract}
We study assignment games in which jobs select machines, and in which certain pairs of jobs may conflict, which is to say they may incur an additional cost when they are both assigned to the same machine, beyond that associated with the increase in load. Questions regarding such interactions apply beyond allocating jobs to machines: when people in a social network choose to align themselves with a group or party, they typically do so based upon not only the inherent quality of that group, but also who amongst their friends (or enemies) choose that group as well. We show how {\em semi-smoothness}, a recently introduced generalization of smoothness, is necessary to find tight bounds on the price of total anarchy, and thus on the quality of correlated and Nash equilibria, for several natural job-assignment games with interacting jobs. For most cases, our bounds on the price of total anarchy are either exactly 2 or approach 2. We also prove new convergence results implied by semi-smoothness for our games. Finally we consider coalitional deviations, and prove results about the existence and quality of strong equilibrium.
\end{abstract}

\section{Introduction}

In many real-world scheduling problems, the cost of assigning a given job to a machine depends not only on the total load of that machine, but also on the particular other jobs assigned there.  Certain pairs of jobs may conflict, which is to say they may incur an additional cost when they are both assigned to the same machine, beyond that associated with the increase in load.  Such conflicts may arise, e.g., when two jobs both place intensive demands on the same computational resources, such as memory, CPU, or network access -- assigning one CPU-intensive job and one memory-intensive job to a single machine is often preferable to assigning two jobs that are both memory-intensive, as the former allows for a more uniform utilization of resources.  Conversely, some pairs of jobs may benefit from being assigned to the same machine.  For example, consider two related processes that need to exchange data during their execution.  By assigning these to a single machine, the additional communication cost of sending messages over a network may be avoided.

Load balancing -- the efficient allocation of jobs to machines -- is one of the most fundamental algorithmic problems in the areas of distributed computing and network design.  The associated classes of games, in which players assign their own jobs to machines, are among the most well-studied within the field of algorithmic game theory (see \cite{CFK+06, CK05, KSTZ04, V07, STZ04} and their references for surveys of this expansive literature). Nevertheless, interactions between jobs like the ones described above have not been extensively analyzed, even for very simple contexts. Questions regarding such interactions apply beyond allocating jobs to machines: when people in a social network choose to align themselves with a group or party, they typically do so based upon not only the inherent quality of that group, but also on who amongst their friends (or enemies) chooses that group as well.
In this setting groups correspond to machines and people correspond to jobs. Although we will usually talk about jobs and machines in this paper, the questions we ask can be framed in the context of social networks as well.

In this paper we aim to model and analyze the effects of such interactions.  To this end, we consider two basic job-assignment games, and superimpose a graph structure on jobs that represents their pairwise interactions.  In the resulting games, a given player's utility depends both on total load of their selected machine and the placement of jobs by their neighbors in this interaction graph.  Despite the simplicity of the underlying models, the addition of this structure yields rich classes of games, fundamentally distinct from the originals, with new properties and non-trivial behavior.  Our primary goal is to understand and quantify the efficiency of these games.

Perhaps the most widely studied characterization of game efficiency is the price of anarchy \cite{KP99} -- the ratio of the social cost of the worst Nash equilibrium to that of an optimal, centrally designed solution.  As the field of algorithmic game theory has matured, more attention has been paid to the evaluation of additional equilibrium concepts. Some of the most notable of these concepts are correlated and coarse-correlated equilibrium, the latter being especially important due to its connection with outcomes of no-regret learning \cite{R09}.
The ratio between the quality of the worst coarse-correlated equilibrium and the centrally optimal solution is sometimes called the {\em price of total anarchy} \cite{BHLR08, R09} (this is how we will use it in this paper, although this concept was originally defined as a bound on the results of no-regret learning). The price of total anarchy trivially acts as a bound on the quality of correlated and Nash equilibrium as well.
In the unifying paper of Roughgarden \cite{R09}, the concept of smoothness is introduced, which provides a general framework for bounding the price of total anarchy, as well as for generating convergence results for best-response dynamics. However, as we show below, the standard version of this smoothness framework cannot yield tight bounds on the quality of equilibria for our games.

We instead show how {\em semi-smoothness}, a generalization of smoothness recently introduced in \cite{LL11} for the analysis of a GSP game, is necessary to form such tight bounds. In particular, semi-smoothness with respect to a mixed strategy is needed to show a tight bound in our games. We also prove new convergence results implied by semi-smoothness for our games. Our results illustrate the usefulness of semi-smoothness as a proof technique for natural games beyond GSP.


\subsection{Our Contributions}
We use semi-smoothness to bound the quality of coarse-correlated equilibrium (CCE) for several natural job-assignment games with interacting jobs.
We show that these bounds are tight, not only for CCE but for correlated and mixed Nash equilibria as well. We also prove new convergence results for best-response dynamics in these games. Specifically, we consider the following games. 

\begin{table}
\centering
\def\arraystretch{1.5}
\begin{tabular}{| r | c|}
\hline 
& Price of Total Anarchy \\
\hline\hline
\minuncut{} &  $2 + \frac{m-1}{n} - \frac{1}{m}$ \\ 
\hline
\mincut{} & $2 - \frac{1}{m}$ \\
\hline
\mincorrelation{} & $1 + \frac{m-1}{n} + \frac{\beta}{\alpha} \frac{m-1}{m} + \frac{\gamma}{\alpha}\left(\frac{m-1}{m} - \frac{m-1}{n} \right) ^{(*)}$  \\
\hline
\maxcut & 2 \\
\hline
\maxuncut{} & $m$ \\
\hline
\end{tabular}
\caption{This table provides a summary of the price of total anarchy upper bounds when $n \geq m$ that follow from semi-smoothness arguments. In all cases, there exist mixed Nash equilibria (pure Nash equilibria for some models) for which these bounds are tight. $\left(^*\right)$ This bound holds for the case in which $\alpha \geq \gamma$. When $\alpha < \gamma$, a weaker bound of $1 + \frac{\beta}{\alpha} \frac{m-1}{m} + \frac{\gamma}{\alpha}\frac{m-1}{m}$ can be derived instead.}
\end{table}

%

In the game that we call {\em Balancing with Conflicts} (BwC), the strategy space of each of the $n$ players consists of $\{1,\ldots,m\}$. We will typically assume $n \geq m$. We can think of each player as a job choosing among $m$ machines, or as a person choosing among $m$ groups. In the usual (and thoroughly analyzed) load balancing game, the cost to a player choosing machine $k$ is an increasing function of $|X_k|$, where $X_k$ is the set of players who choose machine $k$. 
In the BwC game we are also given an undirected graph $G=(V,E)$ in which nodes correspond to players and an edge represents a conflict between these players.
Two players who are neighbors in $G$ and are assigned to the same machine both incur an additional cost.  Specifically,
the cost of a player $i$ who selects strategy $k$ is $$c_i=|X_k|+|X_k\cap\{j|(i,j)\in E\}|.$$
We define the cost of a solution as the sum of player costs, i.e., $c(\vec{s}) = \sum_{i} c_i(\vec{s})$. Even in this simple setting, the addition of conflict costs drastically changes the game. Using semi-smoothness, we show that all CCE are within a factor of $2 - \frac{1}{m} + \frac{m-1}{n}$ ($\frac{3}{2} + \frac{1}{n}$ for 2 machines) of the optimum solution, as well as provide new convergence results. Note that this bound cannot be obtained using arguments based on smoothness. 
In addition to various extensions of this game, we prove that strong Nash equilibrium, which is a solution that is stable with respect to coalitional deviations \cite{NRTV07}, always exists for BwC, and that the strong price of anarchy is at most $\frac{4}{3} + \frac{2}{3n}$ for $m=2$.
Note that games in which people partition themselves into two groups while trying to avoid being with others they dislike is already a nontrivial and important special case \cite{GM09,BCK10}.

We also derive similar results for the version of this game where players attempt to minimize the number of neighbors {\em not} in their group (the player cost is $|X_k|+|\{j|(i,j)\in E\}\backslash X_k|$), i.e., players prefer to group with as many ``friends'' as possible. We call this game \emph{\mincut{}} (BwF). In particular, we are able to show that the price of total anarchy of \mincut{} is $2 - \frac{1}{m}$. 

We observe that \minuncut{} is essentially a ``sum'' of a load balancing game and a minimum uncut game. Similarly, \mincut{} is a ``sum'' of a load balancing game and a minimum cut game. However, we are not limited to summing only two games at a time; we can allow players to simultaneously have conflicts they want to avoid and friends they want to group with. Furthermore, we can also consider weighted sums of these games to reflect the idea that, for example, avoiding conflicts is more crucial than grouping with friends. We provide a generalized model known as \emph{\mincorrelation{}} with a graph $G = (N, E^+, E^-)$, in which $E^+$ and $E^-$ are two disjoint sets of edges representing friendship and conflict, respectively. Then the cost of a player is defined as
$$c_i=\alpha \cdot |X_k|+ \beta \cdot |X_k\cap\{j|(i,j)\in E^-\}| + \gamma \cdot |\{j|(i,j)\in E^+\}\setminus X_{k}|,$$
 where $\alpha >0, \beta, \gamma \geq 0$ are the weights of each game. We provide a price of total anarchy bound of $1 + \frac{m-1}{n} + \frac{\beta}{\alpha} \frac{m-1}{m} + \frac{\gamma}{\alpha}\left(\frac{m-1}{m} - \frac{m-1}{n} \right)$ when $\alpha \geq \gamma$ (and a slightly weaker bound when $\alpha < \gamma$) using semi-smoothness.


We also study a utility-maximizing game which we call {\em Sharing with Conflicts} (SwC). In this game the strategy space of each of the $n$ players still consists of $\{1,\ldots,m\}$, but each machine/group $k$ has an intrinsic value (or processing power) $p_k$. This value is equally shared among the players who choose machine $k$. Without conflicts, this is a variation on the market sharing game of, e.g., \cite{AAE+08, GMT06, MV04}, but once again we suppose that players benefit from avoiding those machines selected by their neighbors in some network $G$. We define the utility of a player $i$ who selects strategy $k$ to be $u_i=p_k/|X_k|+|\{j|(i,j)\in E\}\backslash X_k|.$ We use the semi-smoothness framework to show new convergence results for this and related games, as well as a tight bound of 2 for the quality of coarse correlated equilibrium as compared to the optimum solution. These results also hold for weighted edges.


\subsection{Related Work}

Our games are related to many well-studied games.
There is a significant amount of literature on the subject of selfish load balancing and, more generally, atomic congestion games; consult, for example, \cite{CFK+06, CK05, KSTZ04, V07, STZ04} and their references. In the market sharing game \cite{AAE+08, GMT06, MV04}, there are locations with values that are divided equally among all players at that particular location. Cut games, such as max cut, are represented by graphs in which the nodes correspond to players and the edges signify that two players are either friends or enemies. Although our games are very related (see Section \ref{Conclusion and Future Work}), our results are not simple extensions of known results about cut games. Hoefer \cite{H07} gives price of anarchy bounds for the max $k$-cut game. Gourv\`{e}s and Monnot \cite{GM09, GM10} explored the existence and inefficiency of strong Nash equilibria for the max $k$-cut game. We use techniques similar to \cite{GM10} to prove strong price of anarchy bounds for our games. Since max cut is PLS-complete, efforts have been focused on finding good solutions. For example, Bhalgat et al. \cite{BCK10} give an algorithm for finding approximate Nash equilibrium and Christodoulou et al.  \cite{CMS06} show that random walks can converge to good approximate solutions very quickly.

There are many games that model the idea of players wanting to be close to their friends and away from their enemies, which is closely related to the motivation behind our games. The aforementioned cut games model this concept. Clustering games, such as those found in \cite{FLN12, H07}, are clearly representative of this idea as well. Coordination games in which players choose whether or not to adopt a new technology depending on what their friends are doing \cite{JY07, KKT03, M00, K07} are also related to this concept, since players are essentially choosing to be in the same group with as many of their friends as possible.  However, despite the similarity in motivation, the models and issues arising in this line of work are very different from ours.

The concept of smoothness was originally defined by Roughgarden \cite{R09}. Smoothness can be used to bound the price of anarchy for many equilibrium concepts, and it implies fast convergence to good solutions. Several variants of smoothness have been introduced \cite{CKK+12, LL11, RS11} to address the inability of the original smoothness definition to provide tight price of anarchy bounds for various games. We use the variant in \cite{CKK+12}, known as \emph{semi-smoothness}, for this reason. The concept of $\beta$-niceness was originally defined by Awerbuch, et al. \cite{AAE+08} and explored again in \cite{ACE+11}. Like smoothness, $\beta$-niceness provides bounds for price of anarchy, but only for pure Nash equilibria. Furthermore, a $\beta$-nice potential game (see \cite{NRTV07} for the definition of ``potential games'') and satisfies some other conditions will converge to states that are approximately optimal (although not necessarily stable) in a polynomial number of steps \cite{AAE+08, ACE+11}. Other papers, such as \cite{FFM08, FM09, R09}, show that potential games can converge to good states quickly under certain conditions. We demonstrate that a more general variant of $\beta$-niceness, called \emph{$(\lambda, \mu)$-niceness}, is simply a weaker version of smoothness. Our convergence results generalize those from \cite{ACE+11, R09}; we only require that they be potential games and $(\lambda, \mu)$-nice. The games do not need to satisfy a bounded-jump condition like in \cite{AAE+08, ACE+11}, or be smooth like in \cite{R09}.


\section{Preliminaries and Smoothness Concepts}

In this section we present how semi-smoothness relates to various similar concepts, and establish some of its implications. We begin by defining these concepts for payoff-maximization games; the definitions for cost-minimization games are similar, and are included in Section \ref{Cost-Minimization Games} for completeness. 

\subsection{Payoff-Maximization Games}
\label{Payoff-Maximization Games}

A \emph{payoff-maximization game} is defined by a set of players $N = \{1, 2, \dots, n \}$ where each player $i$ has a set of strategies $S_{i}$, and a payoff function $u_{i} : \times_{i \in N} S_{i} \to \mathbb{R}_{\geq 0}$. We call $\vec{s} \in \times_{i \in N} S_{i}$ a state, outcome, or solution interchangeably. Finally, there is a global objective function $u(\vec{s})$ that will be defined as $u(\vec{s}) = \sum_{i \in N} u_{i}(\vec{s})$. Throughout this paper, $\vec{s}^*$ will refer to the solution that maximizes $u(\vec{s})$.

\textbf{Equilibrium Concepts.} A solution $\vec{s}$ is a \emph{pure Nash equilibrium} if for every player $i$ and for every strategy $s_{i}' \in S_{i}$, $u_{i}(\vec{s}) \geq u_{i}(s_{i}', s_{-i})$. Let $\sigma_{i}$ denote a probability distribution over the strategies of player $i$. A probability distribution $\sigma = \sigma_{1} \times \sigma_{2} \times \dots \times \sigma_{n}$ over solutions is a \emph{mixed Nash equilibrium} if for every player $i$ and every strategy $s_{i}'$, $\mathbf{E}_{s \sim \sigma} \left[u_{i}(\vec{s}) \right] \geq \mathbf{E}_{s \sim \sigma} \left[u_{i}(s_{i}', s_{-i}) \right]$. A probability distribution $\sigma$ over solutions is a \emph{correlated equilibrium} if for every player $i$ and every strategy $s_{i}' \in S_{i}$, $\mathbf{E}_{s \sim \sigma} \left[u_{i}(\vec{s}) | s_{i} \right] \geq \mathbf{E}_{s \sim \sigma} \left[u_{i}(s_{i}', s_{-i} | s_{i}) \right]$. A probability distribution $\sigma$ over solutions is a \emph{coarse correlated equilibrium} if for every player $i$ and every strategy $s_{i}' \in S_{i}$, $\mathbf{E}_{s \sim \sigma} \left[u_{i}(\vec{s}) \right] \geq \mathbf{E}_{s \sim \sigma} \left[u_{i}(s_{i}', s_{-i}) \right]$. A solution $\vec{s}$ is a \emph{strong Nash equilibrium} if for every non-empty subset of players $C \subseteq N$ and for every $s_{C}' \in \times_{i \in C} S_{i} $, there exists a player $i \in C$ such that $u_{i}(\vec{s}) \geq u_{i}(s_{C}', s_{-C})$.

Next we define the \emph{price of anarchy}, which is the measure of the efficiency of equilibria used throughout this paper, and related concepts.

\textbf{Price of Anarchy.} Let $\mathcal{N}(\mathcal{G})$ denote the set of pure Nash equilibrium of a game $\mathcal{G}.$ The \emph{pure price of anarchy} of  $\mathcal{G}$ is $\frac{u(\vec{s}^{*})}{\min_{\vec{s} \in \mathcal{N}(\mathcal{G})} u(\vec{s})}.$ We can analogously define the price of total anarchy and strong price of anarchy as the ratio between the optimal solution and the worst coarse correlated equilibrium and strong Nash equilibrium, respectively. The \emph{price of stability} is a related concept that is the measure between the optimal solution and the best Nash equilibrium. Like the price of anarchy, it can be defined with respect to different equilibrium concepts. The \emph{pure price of stability} of $\mathcal{G}$ is $\frac{u(\vec{s}^{*})}{\max_{\vec{s} \in \mathcal{N}(\mathcal{G})} u(\vec{s})}$.

Recall the classic definition of smoothness from \cite{R09}. Smoothness of a game implies many positive results about it, including price of anarchy bounds for various equilibrium concepts and good convergence results. It acts as a measure of the inefficiency caused by the selfishness of the players.

\begin{mydef} \textbf{(Smooth Games)}
A payoff-maximization game is \emph{$\left(\lambda, \mu \right)$-smooth} if for every pair of outcomes $\vec{s}$ and $\vec{s}'$,
\begin{equation*}
\sum_{i \in N} u_i(s_i', s_{-i}) \geq \lambda \cdot u(\vec{s}') - \mu \cdot u(\vec{s}).
\end{equation*}
\end{mydef}

In the above definition, $u_i(\vec{s})$ is the utility of player $i$  in outcome $\vec{s}$, and $u(\vec{s})=\sum_i u_i(\vec{s})$. Most relevant to our work is the fact that a game being $\left(\lambda, \mu \right)$-smooth with nonnegative $\lambda, \mu$ implies that the total utility of all coarse correlated equilibria (CCE) is at least $\frac{\lambda}{1+\mu}$ times the utility of the optimum solution (the outcome with the highest total utility). This immediately implies a similar bound for outcomes that result from no-regret learning, and for both pure and mixed Nash equilibrium \cite{R09}. Thus, the price of total anarchy \cite{BHLR08} is at most $\frac{1+\mu}{\lambda}$.

Other, more general notions of smoothness have also been studied \cite{CKK+12, LL11, RS11}. Here we mention the one most relevant to our work, known as semi-smoothness.

\begin{mydef} \textbf{(\sseg)}
A payoff-maximization game is \emph{$\left(\lambda, \mu \right)$-\sse} if there exists a mixed strategy $\sigma_i$ for each player $i$ such that for every outcome $\vec{s}$,
\begin{equation*}
\mathbf{E}_\sigma\left[\sum_{i \in N} u_{i}(\sigma_i, s_{-i}) \right] \geq \lambda \cdot u(\vec{s}^{*}) - \mu \cdot u(\vec{s}).
\end{equation*}
\end{mydef}

It is easy to show that a game being semi-smooth implies all the same results from \cite{R09} as smoothness, including a bound of $\frac{1+\mu}{\lambda}$ for the price of total anarchy. Notice that it essentially removes one of the `$\forall$' quantifiers: while smoothness requires that the inequality holds for all $\vec{s}, \vec{s}'$, semi-smoothness only requires that there {\em exist} a mixed outcome $\sigma$ such that for all $\vec{s}$, the inequality holds. Thus, semi-smoothness is quite a bit more general.


Semi-smoothness was first defined in \cite{LL11}, although in \cite{LL11} $\sigma$ was restricted to pure strategies. It was later re-defined with mixed strategies in \cite{CKK+12}. The strategies $\sigma_i$ being mixed makes a large difference when establishing bounds on the price of total anarchy. While in the later sections of this paper we study games where this provides tight bounds on the price of total anarchy, here we give a simple example of a well-known game where semi-smoothness with mixed strategies provides a tight bound, while the standard definition of smoothness (or even semi-smoothness with $\sigma_i$ being pure strategies) does not. Specifically, we consider the well-studied Max-Cut game:

\begin{mydef} \textbf{(Max-Cut Game)} We are given a graph $G = (V, E)$; each player corresponds to a node in $G$. Each player must choose to be in partition 1 or partition 2. The utility of player $i$ is the number of edges between $i$ and players that choose a different partition from $i$.
\end{mydef}

It is well-known that the price of anarchy of the Max-Cut Game is 2. However, as we show below, there is no choice of $\lambda$ and $\mu$ for which this game is $(\lambda,\mu)$-smooth (or semi-smooth with pure strategies) for any $\lambda, \mu$ such that  $\frac{\lambda}{1+\mu}>\frac{1}{3}$. On the other hand, the Max-Cut game is $(\frac{1}{2},0)$-semi-smooth, thus implying not only that the price of anarchy is 2, but also that the utility of any CCE is within a factor of 2 of the optimum. This demonstrates how previously defined smoothness concepts can be inadequate for obtaining tight bounds on the price of total anarchy even for simple games.


\begin{thm}
The Max-Cut Game is $(\frac{1}{2},0)$-semi-smooth. On the other hand, there is no choice of $\lambda$ and $\mu$ for which this game is $(\lambda,\mu)$-smooth (or semi-smooth with pure strategies) such that $\frac{\lambda}{1+\mu}>\frac{1}{3}$.
\end{thm}
\begin{proof}
Let a graph $G = (V, E)$ be an instance of the max-cut game, and let $\sigma_i$ be the mixed strategy where player $i$ chooses each of the two partitions with probability $\frac{1}{2}$. For all outcomes $\vec{s}$, $u_i(1, s_{-i}) + u_i(2, s_{-i})$ equals the total degree of node $i$, which we denote by $d_i$. Thus, we observe that
\begin{align*}
\mathbf{E}\left[\sum_{i \in N} u_{i}(\sigma_{i}, s_{-i}) \right]
= \frac{1}{2} \sum_{i \in N} d_i
= |E|
\geq \frac{1}{2}u(\vec{s}^{*}).
\end{align*}
Thus, the max-cut game is \emph{$\left(\frac{1}{2}, 0 \right)$-\sse}.


Now we show that the same bound cannot be proven using previously studied concepts of smoothness. Consider a graph $G$ which contains two nodes $i, j$ and a single edge between them. Let $\vec{s}^1$ be an optimal solution, i.e., $s_{i}^1 = 1, s_{j}^1 = 2$. Then $u(\vec{s}^1) = 2$. Let $\vec{s}^2$ be a permutation of $\vec{s}^1$, i.e.,  $s_{i}^2 = 2, s_{j}^2 = 1$; $u(\vec{s}^2) = 2$. Let $\vec{s}^3$ be the solution in which both players are in partition 1 and $\vec{s}^4$ be the solution in which both players are in partition 2. Then $u(\vec{s}^3) = u(\vec{s}^4) = 0$. We must show that for any choice of $\sigma$ such that $\sigma$ is a vector of pure strategies, max cut is not $(\lambda, \mu)$-semi-smooth with pure strategies for any $\lambda, \mu$ such that $\frac{\lambda}{1+\mu} > \frac{1}{3}$.

Let $\sigma = \vec{s}^1$. Consider when players deviate from $\vec{s}^2$. If player $i$ deviates from $s_{i}^{2}$ to $s_{i}^{1}$ without player $j$ changing his strategy from $s_{j}^{2}$, then player $i$'s utility becomes 0. Similarly, if player $j$ deviates from $s_{j}^{2}$ to $s_{j}^{1}$, then player $j$'s utility becomes 0. That is, $u_{i}(\sigma_{i}, s_{-i}^{2}) + u_{j}(\sigma_{j}, s_{-j}^{2}) = 0$. We observe that $0 \geq \lambda \cdot u(\vec{s}^{1}) - \mu \cdot u(\vec{s}^{1}) = 2(\lambda - \mu)$ iff $\lambda \leq \mu$. Next, we consider when players deviate from $\vec{s}^{3}$. Player $i$'s utility remains 0 if she deviates to $\vec{s}_{i}^{1}$ from $\vec{s}_{i}^{3}$ since $\vec{s}_{i}^{1} = \vec{s}_{i}^{3}$. However, if player $j$ deviates to $\vec{s}_{j}^{1}$ from $\vec{s}_{j}^{3}$, then her utility becomes 1. That is, $u_{i}(\sigma_{i}, s_{-i}^{3}) + u_{j}(\sigma_{j}, s_{-j}^{3}) = 1$. We observe that $1 \geq  \lambda \cdot u(\vec{s}^{1}) - \mu \cdot u(\vec{s}^{3}) $ iff $\lambda \leq \frac{1}{2}$. Thus, when $\sigma = \vec{s}^1$, $\frac{\lambda}{1 + \mu} \leq \frac{1}{3}$, since  $\lambda \geq 0$, $\lambda \leq \mu$ and $\lambda \leq \frac{1}{2}$. Similar analysis holds for $\sigma = \vec{s}^2$.

Let $\sigma = \vec{s}^3$. Consider when players deviate from $\vec{s}^{3}$. Since the players are deviating to the same strategy in which their utilities were 0, their utilities remain 0. That is, $u_{i}(\sigma_{i}, s_{-i}^{3}) + u_{j}(\sigma_{j}, s_{-j}^{3}) = 0$. We observe that $0 \ngeq \lambda \cdot u(\vec{s}^{1}) - \mu \cdot u(\vec{s}^{3}) $ for any choice of nonnegative $\lambda, \mu$. Similar analysis holds for $\sigma = \vec{s}^4$.

We have considered every possible $\sigma$. We conclude that max cut is not $(\lambda, \mu)$-semi-smooth with pure strategies for any $\lambda, \mu$ such that $\frac{\lambda}{1+\mu} > \frac{1}{3}$. From this, we can also conclude that max cut is not $(\lambda, \mu)$-smooth for any $\lambda, \mu$ such that $\frac{\lambda}{1 + \mu} > \frac{1}{3}$.
\end{proof}

Finally, we introduce another related notion of smoothness. This last definition of smoothness is actually a new, generalized version of $\beta$\emph{-niceness} from \cite{ACE+11} and \cite{AAE+08}. It is much weaker than the previous versions; while semi-smoothness still implies all the same results as the usual concept of smoothness (including convergence results, as we discuss below), niceness only implies bounds for Nash equilibrium and price of anarchy, and says nothing about CCE or price of total anarchy. However, it is sufficient to provide good convergence results.

\begin{mydef} \textbf{(Nice Games)}
A payoff-maximization game is \emph{$(\lambda, \mu)$-nice} if for every outcome $\vec{s}$, there exists
an outcome $\vec{s}'$ such that
\begin{equation*}
\sum_{i \in N} u_i(s_i', s_{-i}) \geq \lambda \cdot u(\vec{s}^*) - \mu \cdot u(\vec{s}).
\end{equation*}
\end{mydef}

Recall the definition of $\beta$-niceness from \cite{ACE+11}: a payoff-maximization game is $\beta$-nice if for every outcome $\vec{s}$, $\beta \cdot u(\vec{s}) + \sum_{i \in N} \left(u_{i}(s_{i}^{b}, s_{-i}) - u_{i}(\vec{s}) \right) \geq u(\vec{s}^{*})$, where $s_{i}^{b}$ denotes the best response of player $i$ in state $\vec{s}$. Rearranging the terms demonstrates that if a game is $\beta$-nice, then it is $(\lambda, \mu)$-nice with $\lambda = 1, \mu = \beta - 1$, and $\forall \vec{s}$, $\vec{s}' = \left(s_{1}^{b}, \dots, s_{n}^{b} \right)$. In fact, if a game is $(\lambda, \mu)$-nice, then for every state $\vec{s}$, we can always define $\vec{s}'$ to be the vector of best responses $\left(s_{1}^{b}, \dots, s_{n}^{b} \right)$, because $\sum_{i \in N} u_{i}(s_{i}^{b}, s_{-i}) \geq \sum_{i \in N} u_{i}(s_{i}', s_{-i})$ for any other choice of $\vec{s}'$. We frame $\beta$-niceness as $(\lambda, \mu)$-niceness to allow us to prove our convergence results more easily and to make it clear that niceness is a generalization of smoothness.

The difference between nice and semi-smooth games is that $\vec{s}'$ is allowed to depend on $\vec{s}$, while in semi-smooth games $\sigma$ must be a single mixed state, independent of $\vec{s}$. It is easy to verify that while niceness implies a bound of $\frac{1+\mu}{\lambda}$ on price of anarchy, to show the same results for CCE or no-regret outcomes, the solution that the players are ``switching to" must be independent of $\vec{s}$. It follows from their definitions that for fixed $\lambda, \mu$, \emph{$\left(\lambda, \mu \right)$-smooth games} $\subseteq$ \emph{$\left(\lambda, \mu \right)$-\sse} \emph{games} $\subseteq$ \emph{$\left(\lambda, \mu \right)$-nice games}.

\subsection{Convergence Results for Payoff-Maximization Games}
Now that we have defined various concepts related to smoothness, and established which of these concepts imply bounds on the quality of CCE and/or Nash equilibrium, we proceed to describe convergence results, which only require that a game is {\em nice}, and therefore also hold for all the concepts mentioned above. Note that as in \cite{ACE+11, R09}, these results only hold for potential games \cite{MS96}.

We define $\Delta_{i}(\vec{s}) = u_{i}(s_{i}^{b}, s_{-i}) - u_{i}(\vec{s})$, where $s_i^b$ is the best response of player $i$ to state $s$, as the amount of utility player $i$ gains from deviating, and $\Delta(\vec{s}) = \sum_{i \in N} \Delta_{i}(\vec{s})$. When we refer to the BR dynamics, we will mean a sequence of best response moves in which the player selected to make their move is the one with the maximum improvement $\Delta_i(\vec{s})$ \cite{ACE+11}. Similar results can be shown for many other types of best response dynamics.

The first convergence result is a slight generalization of the one given by \cite{R09}, with the proof being essentially the same. It works for more general potential functions instead of being restricted to lower potential functions, as in \cite{R09}. As usual, $\vec{s}^*$ refers to the solution with highest utility.

\begin{thm}
\label{Utility Convergence 1}
Consider any payoff-maximization game $H$ that has an exact potential function $\Phi(\vec{s})$ such that for some $B\geq 1$, we have that $u(\vec{s}) \geq \frac{1}{B}\Phi(\vec{s})$ and is $(\lambda, \mu)$\emph{-nice}. Let $\rho=\frac{\lambda}{1+\mu}$. Then, for any $\epsilon > 0$, the BR dynamic reaches a state $\vec{s}^{t}$ with $u(\vec{s}^{t}) \geq \frac{\rho}{1+\epsilon} \cdot u(\vec{s}^{*})$ from any starting state $\vec{s}^{0}$ in at most $O\left(\frac{Bn }{\epsilon(1+\mu)} \log (B\cdot u(\vec{s}^{*})) \right)$ steps.
\end{thm}


Note that the above theorem does {\em not} say that all states after state $s^t$ will be good, just that a good state will be reached in a small amount of time. Our second convergence result is a generalization of the convergence result for perfect $\beta$\emph{-nice} games in \cite{ACE+11}. Our version of this result does not require the game to be perfect and allows more general potential functions.

\begin{thm}
\label{Utility Convergence 2}
Consider any payoff-maximization game $H$ that has an exact potential function $\Phi(\vec{s})$ such that for some $A,B\geq 1$, we have that $A  \Phi(\vec{s}) \geq u(\vec{s}) \geq \frac{1}{B}\Phi(\vec{s})$ and is $(\lambda, \mu)$\emph{-nice}. Let $\rho=\frac{\lambda}{1+\mu}$. Then, for any $\epsilon > 0$, the BR dynamic reaches a state $\vec{s}^{t}$ with $u(\vec{s}^{t}) \geq \frac{\rho(1-\epsilon)}{AB} \cdot u(\vec{s}^{*})$  in at most $O\left(\frac{n}{A(1+\mu)} \log \frac{1}{\epsilon} \right)$ steps from any initial state. Furthermore, all future states reached with best-response dynamics will satisfy this approximation factor as well.
\end{thm}
\begin{proof}
Let $i$ be the player selected by the BR dynamic at any given time step. By our hypotheses about $\Phi(\vec{s})$ and the definition of $(\lambda, \mu)$\emph{-nice},
\begin{align*}
\Phi(s_{i}^{b}, s_{-i}) - \Phi(\vec{s}) &= u_{i}(s_{i}^{b}, s_{-i}) - u_{i}(\vec{s}) \\
&\geq \frac{\Delta(\vec{s})}{n} \\
&\geq \frac{\lambda \cdot u(\vec{s}^{*}) - (1 + \mu) \cdot u(\vec{s})}{n} \\
&\geq \frac{\lambda \cdot u(\vec{s}^{*}) - A(1 + \mu) \Phi(\vec{s})}{n}.
\end{align*}
Let $f(\vec{s}) = \frac{\lambda \cdot u(\vec{s}^{*}) - A(1 + \mu) \Phi(\vec{s})}{n}$. Then $f(\vec{s}) - f(s_{i}^{b}, s_{-i}) = \frac{A(1 + \mu)}{n} \left(\Phi(s_{i}^{b}, s_{-i}) - \Phi(\vec{s}) \right) \geq \frac{A(1 + \mu)}{n}f(\vec{s})$, which implies that $f(s_{i}^{b}, s_{-i}) \leq \left(1 - \frac{A(1 + \mu)}{n}\right) f(\vec{s})$. Thus, from a starting state $\vec{s}^0$, the BR dynamic converges to a state $\vec{s}^{t}$ with
\begin{equation*}
f(\vec{s}^{t}) \leq \left(1 - \frac{A(1 + \mu)}{n}\right)^{t} f(\vec{s}^{0}).
\end{equation*}
We can set $t = \lceil \frac{n}{A(1 + \mu)} \ln \frac{1}{\epsilon} \rceil $. Using the fact that $\left(1 - \frac{1}{x} \right)^{x} \leq 1/e$, we are able to derive that $f(\vec{s}^{t}) \leq e^{\ln \epsilon^{-1}} f(\vec{s}^0) \leq \epsilon \cdot f(\vec{s}^{0})  \leq \frac{\epsilon \lambda \cdot u(\vec{s}^{*})}{n}.$ Finally, using these results and our $\Phi(\vec{s})$ bounds, we can derive
\begin{align*}
u(\vec{s}^{t}) &\geq \frac{1}{B} \Phi(\vec{s}^{t}) \\
&= \frac{n}{AB(1+\mu)} \left(\frac{\lambda \cdot u(\vec{s}^{*})}{n} - f(\vec{s}^{t}) \right) \\
&\geq \frac{n}{AB(1+\mu)}\left((1 - \epsilon)\frac{\lambda \cdot u(\vec{s}^{*})}{n} \right) \\
&\geq \frac{\rho(1-\epsilon)}{AB} \cdot u(\vec{s}^{*}).
\end{align*}
Because $\Phi(\vec{s}^{t}) \geq \frac{\rho(1-\epsilon)}{A} \cdot u(\vec{s}^{*})$, we know that this approximation factor will hold for all future states reached by best-response dynamics as well, since $\Phi$ will only increase.
\end{proof}

%

\subsection{Cost-Minimization Games}
\label{Cost-Minimization Games}
While above we defined everything for utility-maximization games, the same concepts can be defined and results hold for cost-minimization games. For completeness, we define these concepts here.

\begin{mydef} \textbf{(Smooth Games)}
A cost-minimization game is \emph{$\left(\lambda, \mu \right)$-smooth} if for every pair of outcomes $\vec{s}$ and $\vec{s}'$,
\begin{equation*}
\sum_{i \in N} c_i(s_i', s_{-i}) \leq \lambda \cdot c(\vec{s}') + \mu \cdot c(\vec{s}).
\end{equation*}
\end{mydef}


\begin{mydef} \textbf{(\sseg)}
A cost-minimization game is \emph{$\left(\lambda, \mu \right)$-\sse} if there exists a mixed strategy $\sigma_i$ for each player $i$ such that for every outcome $\vec{s}$,
\begin{equation*}
\mathbf{E}_\sigma\left[\sum_{i \in N} c_{i}(\sigma_i, s_{-i}) \right] \leq \lambda \cdot c(\vec{s}^{*}) + \mu \cdot c(\vec{s}).
\end{equation*}
\end{mydef}

If a cost-minimization game is \emph{$\left(\lambda, \mu \right)$-\sse}, then the price of anarchy for coarse correlated equilibria is at most $\frac{\lambda}{1 - \mu}$.

\begin{mydef} \textbf{(Nice Games)}
A cost-minimization game is \emph{$\left(\lambda, \mu \right)$-nice} if for every outcome $\vec{s}$, there exists an outcome $\vec{s}'$ such that
\begin{equation*}
\sum_{i \in N} c_i(s_i', s_{-i}) \leq \lambda \cdot c(\vec{s}^{*}) + \mu \cdot c(\vec{s}).
\end{equation*}
\end{mydef}

\section{Combinations of Cost-Minimization Games}
\label{sec:costs}



\subsection{\minuncut{}}

\textbf{Definition.} \minuncut{} can be described as a triple $(N, M, G)$, where $N = \{1, 2, \ldots, n \}$ is a set of players, $M = \{1, 2, \ldots, m \}$ is a set of machines, and $G = (N, E)$ is an undirected graph on the players. Each player $i$ must select one machine. That is, the strategy set for each player $i$ is $S_i = M$. Intuitively, this models a job being assigned to a machine to be processed.
A state of our game is the assignment of every player to a machine, which can be represented by a vector $\vec{s} = (s_1, s_2, \ldots, s_n )$ where $s_i \in S_i$ for each player $i \in N$.

Given a state $\vec{s}$, let $X_{k}(\vec{s})$ denote the set of players assigned to machine $k$, and let $x_{k}(\vec{s}) = |X_{k}(\vec{s})|$. Let $E(X, Y) = \{(i, j) \in E:  i \in X, j \in Y \}$ denote the edges between the players in sets $X$ and $Y$, and let $e(X, Y) = |E(X, Y)|$. We define $E(X) = E(X,X)$ with $e(X) = |E(X)|$.

We will now define the cost incurred by a player.
For each player, including itself, that is assigned to to the same machine as player $i$, $i$ incurs a cost of 1. Furthermore, if two jobs on the same machine conflict there is an additional cost,
i.e., for each player $j$ such that $j$ shares an edge with player $i$ and $j$ is in the same machine as $i$, $i$ incurs an additional cost of 1. Formally, for a player $i$ with $s_i = k$, its cost is $c_i(\vec{s}) = x_{k}(\vec{s})  + e(\{i\}, X_{k}(\vec{s}))$. Finally, we define the cost of an outcome of our game as the total cost experienced by the players,
that is, $c(\vec{s}) = \sum_{i \in N}  c_{i}(\vec{s})= \sum_{k = 1}^{m} \left(x_{k}(\vec{s})^2 + 2e(X_{k}(\vec{s}))\right)$. 

We will show that \minuncut{} has several desirable properties. First, it is easy to see that \minuncut{} is an exact potential game with $\Phi(\vec{s}) = \frac{1}{2}c(\vec{s})$. This guarantees that every instance of \minuncut{} admits a pure Nash equilibrium. The existence of a potential function guarantees that best-response dynamics will converge to a stable solution, but this may require exponential time. We will show later that other properties guarantee that best-response dynamics converge to a good (but not necessarily stable) state very quickly.

Since $c(\vec{s}^{*})$ is the minimum value of the objective function by definition, it follows that $\Phi(\vec{s}^{*})$ is the global minimum of the potential function. This immediately implies that the optimal solution is a pure Nash equilibrium. Thus, we conclude that the price of stability of the \minuncut{} game is 1.

We will now give the main result of this section, but we will prove it later. Previous smoothness definitions were too strong to find tight price of anarchy bounds for \minuncut{}. However, using \emph{\sseness}, we are able to derive tight price of total anarchy bounds. 

\begin{thm}
\label{Min Uncut SS}
The \minuncut{} game  is \emph{$(2 - \frac{1}{m} + \frac{m-1}{n}, 0)$-\sse} for $n \geq m$ and \emph{$(1 + \frac{2n-2}{m}, 0)$-\sse} for $n < m$.
\end{thm}

The following price of anarchy upper bound results for pure Nash equilibrium, mixed Nash equilibrium, correlated equilibrium, and coarse correlated equilibrium follow immediately from Theorem~\ref{Min Uncut SS}.

\begin{cor}
The price of total anarchy of the \minuncut{} game is at most $2 - \frac{1}{m} + \frac{m-1}{n}$ for $n \geq m$, $1 + \frac{2n-2}{m}$ for $n < m$, and $3- \frac{2}{m}$ for any choice of $m$.
\end{cor}

We observe that when $n$ is asymptotically larger than $m$, the price of total anarchy approaches $2 - \frac{1}{m}$.

We are not able to obtain tight bounds on the pure price of anarchy using semi-smoothness due to the fact that the worst coarse correlated equilibria are not pure Nash equilibria for \minuncut{}. In order to provide improved bounds on the pure price of anarchy, we use $\left(\lambda, \mu \right)$-niceness. Furthermore, this result has implications beyond price of anarchy; it guarantees fast convergence to ``good" solutions. 

\begin{thm}
\label{Min Uncut Pure PoA}
\minuncut{} is $\left(2 - \frac{m}{n},0\right)$-nice. The pure price of anarchy of \minuncut{} is at most $2 - \frac{m}{n}$, and this is tight for all $m \leq n$.
\end{thm}

When $m \geq n$, we observe that the pure price of anarchy is trivially 1. We suspect for $n > m^2$, the actual bound is $2 - \frac{1}{m}$.

We will now consider how to find good solutions. First, we can trivially create a 2-approximately stable solution by randomly assigning players to machines such that at most $\lceil \frac{n}{m} \rceil$ players are assigned to each machine. Secondly, we can use best-response dynamics from any initial state to reach a good state. In \cite{ACE+11}, Augustine et al. show that perfect $\beta$-nice payoff-maximization games reach approximately optimal states very quickly. One can easily derive a similar theorem for cost games that are perfect $(\lambda, \mu)$-nice. A game is \emph{perfect} if for any solution $\vec{s}$ and any improving move $s_{i}'$ of player $i$,
\begin{equation*}
c(\vec{s}) - c(s_{i}', s_{-i}) \geq \Phi(\vec{s}) - \Phi(s_{i}', s_{-i}) \geq c_{i}(\vec{s}) - c_{i}(s_{i}', s_{-i})
\end{equation*}
Since $\Phi(\vec{s}) = \frac{1}{2}c(\vec{s})$ is an exact potential function, it follows that for every job $i$ and strategy $s_{i}'$, $c(\vec{s}) - c(s_{i}', s_{-i}) = 2\left(\Phi(\vec{s})  -  \Phi(s_{i}', s_{-i}) \right) \geq \Phi(\vec{s})  -  \Phi(s_{i}', s_{-i}) = c_{i}(\vec{s}) - c_{i}(s_{i}', s_{-i})$. Since \minuncut{} is perfect $(2-\frac{m}{n}, 0)$-nice, the following corollary follows from Theorem 1 of \cite{ACE+11}.

\begin{cor}
For any instance of the \minuncut{} game and for any $\epsilon > 0$, the BR dynamic reaches a state $\vec{s}^{t}$ with
\begin{equation*} c(\vec{s}^{t}) \leq \left(2 - \frac{m}{n} \right) (1+\epsilon) \cdot c(\vec{s}^{*})
\end{equation*}
in at most
$O\left(n \log \frac{m}{\epsilon} \right)$ steps from any initial state $\vec{s}^{0}$. This approximation factor holds for all further states reached by best-response dynamics.
\end{cor}

We will now focus on proving Theorem~\ref{Min Uncut SS}. We need lower bounds for the optimal solution when $n \geq m$ in order to do so. This first lower bound for all solutions follows from the fact that the load balancing cost is minimized when $\frac{n}{m}$ jobs are assigned to each machine.

\begin{lem}
\label{Min Uncut OPT LB 1}
For any solution $\vec{s}$, $c(\vec{s}) \geq \frac{n^2}{m}$.
\end{lem}
\begin{proof}
Let $\vec{x} = (x_1(\vec{s}), x_2(\vec{s}), \dots, x_m((\vec{s})))$. Let $\|\vec{x} \|_{1}, \|\vec{x} \|_{2}$ denote the $L_{1}$ and $L_{2}$ norms, respectively. It is known that for any $m$ dimensional vector $\vec{x}$, $\|\vec{x} \|_{1} \leq  \sqrt{m} \cdot \|\vec{x} \|_{2}$. Using this property, we can derive
\begin{align*}
\sum_{k=1}^{m} x_{k}(\vec{s})^2 &= \|\vec{x}\|_2^2
\geq \left(\frac{\|\vec{x}\|_1}{\sqrt{m}} \right)^2 = \frac{\left(\sum_{k=1}^{m} x_{k}(\vec{s})\right)^2}{m} = \frac{n^2}{m}.
\end{align*}
Thus, $c(\vec{s}) = \sum_{k=1}^{m} \left(x_{k}(\vec{s})^{2} + 2e(X_{k}(\vec{s}))\right) \geq \frac{n^2}{m}$.
\end{proof}

This next lower bound makes use of the fact that if $G$ has a large enough number of edges, then the solution will unavoidably incur some edge cost.

\begin{lem}
\label{Min Uncut OPT LB 2}
For any solution $\vec{s}$, $(m-1) \cdot c(\vec{s}) \geq 2|E|$.
\end{lem}
\begin{proof}
We begin by considering, for any pair of machines $k, l$ such that $k \neq l$, the number of edges between the set of players on machine $k$ and the set of players on machine $l$. The number of edges between these sets is maximized when every player on machine $k$ shares an edge with every player on machine $l$. That is, $e(X_{k}(\vec{s}), X_{l}(\vec{s})) \leq x_{k}(\vec{s}) \cdot x_{l}(\vec{s})$. Using this fact and Lemma~\ref{Min Uncut OPT LB 1},
\begin{align*}
\sum_{k = 1}^{m-1}  \sum_{l = k+1}^{m} e(X_{k}(\vec{s}), X_{l}(\vec{s}))
&\leq \sum_{k = 1}^{m-1} \sum_{l = k + 1}^{m} x_{k}(\vec{s}) \cdot x_{l}(\vec{s}) \\
&= \frac{1}{2} \left(\sum_{k = 1}^{m} \sum_{l = 1}^{m} x_{k}(\vec{s}) \cdot x_{l}(\vec{s}) - \sum_{k=1}^{m} x_{k}(\vec{s})^2 \right) \\
&= \frac{1}{2} \left(n^2 - \sum_{k=1}^{m} x_{k}(\vec{s})^2  \right) \\
&\leq \frac{1}{2} \left(n^2 - \frac{n^2}{m}  \right) \\
&= \frac{(m-1) \cdot n^2}{2m}.
\end{align*}
The number of edges between players on different machines plus the number of edges between players on the same machines is equal to the total number of edges in the graph $G$. That is, $|E| = \sum_{k=1}^{m-1}\sum_{l = k+1}^{m} e(X_{k}(\vec{s}), X_{l}(\vec{s})) + \sum_{k=1}^{m} e(X_{k}(\vec{s}))$. Using this fact, our previously derived result, and Lemma~\ref{Min Uncut OPT LB 1},
\begin{align*}
c(\vec{s}) &= \sum_{k=1}^{m} \left(x_{k}(\vec{s})^{2} + 2e(X_{k}(\vec{s})) \right) \\
&\geq \frac{n^2}{m} + 2\left(|E| -  \sum_{k = 1}^{m-1} \sum_{l = k + 1}^{m} e(X_{k}(\vec{s}), X_{l}(\vec{s})) \right) \\
&\geq 2|E| + \frac{(2 - m) \cdot n^2}{m} \\
&\geq 2|E| + (2 - m) \cdot c(\vec{s}).
\end{align*}
Rearranging the terms completes the proof.
\end{proof}

We finally have the necessary lower bounds, and we will now proceed with the proof of our primary result.

\begin{proof}[Proof of Theorem \ref{Min Uncut SS}]

Let $\vec{s}$ be an outcome, and let $d_i$ denote the degree of player $i$ in $G$. For each player $i$ with $s_{i} = k$, $c_{i}(\vec{s}) = x_{k}(\vec{s}) + e(\{i \}, X_{k}(\vec{s}))$. If player $i$ deviates to a different machine $l$, then their cost will be the sum of number of players already on machine $l$ plus 1 for itself and the number of edges from $i$ to players on machine $l$. That is, for each $s_{i}' = l \neq s_{i}$, $c_{i}(s_{i}', s_{-i}) = x_{l}(\vec{s}) + 1 + e(\{i \}, X_{l}(\vec{s}))$. For each player $i$, let $\sigma_i$ denote the mixed strategy in which every strategy $s_{i}' \in S_{i}$ is selected with probability $\frac{1}{m}$. Then for every player $i$,
\begin{equation*}
\mathbf{E}\left[ c_i(\sigma_{i}, s_{-i}) \right] = \frac{1}{m} \left( \sum_{k = 1}^{m} \left(x_{k}(\vec{s}) + e(\{i \}, X_{k}(\vec{s})) \right) + m - 1 \right) = \frac{1}{m} \left(n + d_i + m - 1 \right).
\end{equation*}

By linearity of expectation and our previous equation,
\begin{align*}
\mathbf{E} \left[\sum_{i \in N} c_{i}(\sigma_{i}, s_{-i}) \right]
&= \sum_{i \in N} \mathbf{E}\left[ c_i(\sigma_{i}, s_{-i}) \right] \\
&= \frac{1}{m} \sum_{i \in N} \left(n + d_i + m - 1 \right) \\
&= \frac{1}{m} \left(n^2 + 2|E| + (m-1) \cdot n \right) \\
&= \left(1 + \frac{m-1}{n} \right) \frac{n^2}{m} + \frac{2|E|}{m}
\end{align*}
We can use our lower bounds for OPT from Lemmas~\ref{Min Uncut OPT LB 1} and \ref{Min Uncut OPT LB 2} to show that for $n \geq m$ this quantity is at most
\begin{equation*}
\left(2 - \frac{1}{m} + \frac{m-1}{n} \right) \cdot c(\vec{s}^*).
\end{equation*}
In the case that $n < m$, we use the fact that OPT = $n$ and that $|E| \leq \frac{n(n-1)}{2}$ to show that the quantity is at most
\begin{equation*}
\left(1 + \frac{2n-2}{m} \right) \cdot c(\vec{s}^*). \qedhere
\end{equation*}
\end{proof}

Given these upper bounds, we observe that price of total anarchy is maximized when $n = m$, which implies that the price of total anarchy is at most $3 - \frac{2}{m}$ for any choice of $m$.

We now give lower bounds on the price of total anarchy to show that these upper bounds are tight.

\begin{claim}
\label{Min Uncut LB Examples}
The price of total anarchy of \minuncut{} is at least $2 - \frac{1}{m} + \frac{m-1}{n}$ for $n \geq m$ and at least $1 + \frac{2n - 2}{m}$ for $n < m$.
\end{claim}
\begin{proof}
When $n \geq m$, consider the case where there are $m$ machines and $G$ is a complete $m$-partite graph where each of the $m$ disjoint sets has $m$ nodes. Then $n = m^2.$

Consider the probability distribution $\sigma_{i}$ over $S_{i}$ in which each job $i$ is assigned to each machine $k$ with probability $\frac{1}{m}$. We claim the product distribution $\sigma = \times_{i \in N} \sigma_{i}$ is a mixed Nash equilibrium. We observe that for all jobs $i$, $\mathbf{E}_{\vec{s} \sim \sigma} \left[c_{i}(\vec{s}) \right] = \sum_{k \in M} \mathbf{E}_{\vec{s} \sim \sigma}\left[c_{i}(\vec{s}) | s_{i} = k \right] P \left(s_{i} = k \right).$ Let $X_{k}'$ be a random variable whose value is the number of the remaining $n-1$ players assigned to some machine $k$. $\mathbf{E} \left[X_{k}' \right] = \frac{n-1}{m}$. Furthermore, the expected edge cost of job $i$ is the product of $\mathbf{E} \left[X_{k}' \right]$ and the probability that $i$ has an edge to another player, which is $\frac{n - m}{n-1}$. Thus, $\mathbf{E}_{\vec{s} \sim \sigma}\left[c_{i}(\vec{s}) | s_{i} = k \right] = 1 + \left(1 + \frac{n - m}{n - 1} \right)\frac{n-1}{m} = \frac{2n-1}{m}$. It follows that $\mathbf{E}_{\vec{s} \sim \sigma} \left[c_{i}(\vec{s}) \right] = \frac{2n-1}{m}$, which implies that $\sigma$ is a mixed Nash equilibrium since the expected cost of job $i$ deviating to any single machine $k$ is equal to its expected cost in $\sigma$. By linearity of expectation, $\mathbf{E}_{\vec{s} \sim \sigma}\left[c(\vec{s}) \right] = \frac{2n^2 - n}{m}$. Since OPT $= \frac{n^2}{m}$ and $n = m^2$, the mixed price of anarchy is $\frac{2n^2 - n}{m} \cdot \frac{m}{n^2} = 2 - \frac{1}{n} = 2 - \frac{1}{m} + \frac{m-1}{n}$.

When $n < m$, if there are $m$ machines and $G = K_{n}$, then the state where each job $i$ is assigned to each machine $k$ with probability $\frac{1}{m}$ is the worst mixed Nash equilibrium. The analysis is similar to the previous example.
\end{proof}

Our lower bound examples given in Claim~\ref{Min Uncut LB Examples} show that mixed price of anarchy, and as result, the price of total anarchy upper bounds obtained by semi-smoothness are tight. We now provide proof of Theorem~\ref{Min Uncut Pure PoA}, which gives an improved bound on the pure price of anarchy that cannot be obtained using semi-smoothness.

\begin{proof}[Proof of Theorem~\ref{Min Uncut Pure PoA}]
We will assume that $m$ does not divide $n$; the case where $m$ divides $n$ is similar. In every solution, every player's best-response is guaranteed to have load balancing cost at most $\lceil \frac{n}{m} \rceil$ because there is guaranteed to be at least one machine with load at most $\lfloor \frac{n}{m} \rfloor$. In order for a player's best-response to have load balancing cost equal to $\lceil \frac{n}{m} \rceil$, every machine must have load at least $\lfloor \frac{n}{m} \rfloor$ and that player must be on a machine with load at least $\lceil \frac{n}{m} \rceil$. From this, we conclude that the best-responses of at most $(n \mod m) \lceil \frac{n}{m} \rceil$ players have load balancing cost equal to $\lceil \frac{n}{m} \rceil$. The best-responses of the remaining players have load balancing cost at most $\lfloor \frac{n}{m} \rfloor$. Similar to Lemma~\ref{Min Uncut OPT LB 1}, $c(\vec{s}^*) \geq (n \mod m) \lceil \frac{n}{m} \rceil^{2}$ $+ (n - (n \mod m)\lceil \frac{n}{m} \rceil) \lfloor \frac{n}{m} \rfloor$. We conclude that the sum of the load balancing cost of each player's best-response is at most $c(\vec{s}^*)$. 

Since the edge cost of a player is strictly less than load balancing cost, we claim that the edge cost of the best-responses is at most $\lceil \frac{n}{m} \rceil - 1$ for at most $(n \mod m) \lceil \frac{n}{m} \rceil$ players and at most $\lfloor \frac{n}{m} \rfloor - 1$ for the remaining players. We conclude that the sum of the edge cost of each player's best-response is at most $c(\vec{s}^*) - n$. Then for all solutions $\vec{s}$, $\sum_{i \in N} c_{i}(s_{i}^{b}, s_{-i}) \leq 2c(\vec{s}^*) - n \leq (2 - \frac{m}{n}) \cdot c(\vec{s}^*)$ by Lemma~\ref{Min Uncut OPT LB 1}. We conclude that \minuncut{} is $(2 - \frac{m}{n}, 0)$-nice. 

We will now give a lower bound example. Consider the case where there are $m$ machines and $G$ is a complete $m$-partite graph where each of the $m$ disjoint sets has $m$ nodes. Then $n = m^2.$ Solution $\vec{s}$ is formed by assigning the jobs in the first disjoint set to the first machine, the jobs in the second disjoint set to the second machine, etc. Then for each machine $k$, $x_{k}(\vec{s}) = m$. Then the cost of this solution is $m^3 = \frac{n^2}{m}$. Since $c(\vec{s}^{*}) \geq \frac{n^2}{m}$, $\vec{s}$ must be the optimal solution.

Now consider the solution $\vec{s}'$ formed by doing the following: for each disjoint set, assign all of the jobs in that set to different machines. We claim $\vec{s}'$ is a Nash equilibrium. The cost of a job $i$ is $2m-1.$ If that job is reassigned to any other machine, then the cost of $i$ will be $2m$. Thus, $\vec{s}'$ is a Nash equilibrium. We will now calculate $c(\vec{s}')$. For each machine $k$, $x_{k}(\vec{s}') = m$ and $e_{k}(\vec{s}') = \frac{m(m-1)}{2}$. Thus, $c(\vec{s}') = 2m^3 - m^2.$ Since $\frac{2m^3 - m^2}{m^3} = 2 - \frac{m}{n}$, we conclude that the price of anarchy is $\geq 2 - \frac{m}{n}$.
\end{proof}

%
%
%
%

\subsection{\mincut}

%


\textbf{Definition.} We now consider a variation of \minuncut{} in which the additional cost arises from separating players who are ``friends".
 The \mincut{} game is identical to the \minuncut{} except jobs incur additional costs for having edges with jobs that are not assigned to the same machine. That is, for every job $i$ with $s_{i} = k$, $c_{i}(\vec{s}) = x_{k}(\vec{s}) + \sum_{l \neq k} e(\{i\}, X_{l}(\vec{s}))$.

We can use the exact same techniques for the \mincut{} game as we did with the \minuncut{} game to prove better results: we are able to obtain tight price of anarchy bounds for \mincut{}. As with the \minuncut{} game, the \mincut{} game is an exact potential game with $\Phi(\vec{s}) = \frac{1}{2}c(\vec{s})$, which implies that the price of stability is 1.

\begin{thm}
\label{Min Cut SS}
The \mincut{} game is $\left(2 - \frac{1}{m}, 0 \right)$-semi-smooth. The price of total anarchy is at most $2 - \frac{1}{m}$, and it is tight.
\end{thm}
\begin{proof}
We can use the same argument as Theorem~\ref{Min Uncut SS} to show that \mincut{} is semi-smooth if we use the following two lower bounds for OPT: for any solution $\vec{s}$, $c(\vec{s}) \geq \frac{n^2}{m}$, and $c(\vec{s}) \geq 2|E| + n$. The former follows by the same argument that we used to prove Lemma~\ref{Min Uncut OPT LB 1}. We now offer proof of the latter bound.
\begin{align*}
c(\vec{s}) &= \sum_{k=1}^{m} x_{k}(\vec{s})^2 + 2 \sum_{k=1}^{m-1} \sum_{l=k+1}^{m} e\left(X_{k}(\vec{s}), X_{l}(\vec{s}) \right) \\
&= \sum_{k=1}^{m} x_{k}(\vec{s})^2 + 2 \left(|E| - \sum_{k=1}^{m} e\left(X_{k}(\vec{s}) \right) \right) \\
&\geq \sum_{k=1}^{m} x_{k}(\vec{s})^2 + 2|E| - 2 \left(\sum_{k=1}^{m} \frac{x_k(\vec{s}) \left(x_{k}(\vec{s}) - 1\right)}{2} \right) \\
&= 2|E| + \sum_{k=1}^{m} x_{k}(\vec{s}) \\
&= 2|E| + n.
\end{align*}
Furthermore, if we pick the same mixed strategy $\sigma$ as in the proof of Theorem~\ref{Min Uncut SS}, then using similar arguments, we are able to show that
\begin{align*}
\mathbf{E}\left[\sum_{i \in N} c_i(\sigma_{i}, s_{-i}) \right] &= \frac{n^2}{m} \left(1 + \frac{m-1}{n} \right)  + 2\frac{m-1}{m} |E|.
\end{align*} 
Combining this with our lower bounds for the optimal solution completes our proof.
 
To show our price of total anarchy bound is tight, let $G$ consist of $m$ disjoint cliques of size $m$ each. Assigning each clique to the same machine such that only one clique is assigned to each machine is the optimal solution. Assigning each job in the same clique to a different machine is the worst case Nash equilibrium.
\end{proof}

\begin{cor}
For any instance of the \mincut{} and for any $\epsilon > 0$, the BR dynamic reaches a state $\vec{s}^{t}$ with
\begin{equation*} c(\vec{s}^{t}) \leq \left(2 - \frac{1}{m}\right) (1+\epsilon) \cdot c(\vec{s}^{*})
\end{equation*}
in at most 
$O\left(n \log \frac{m}{\epsilon} \right)$
steps from any initial state $\vec{s}^{0}$. This approximation factor holds for all further states reached by best-response dynamics.
\end{cor}

\subsection{\mincorrelation{}} 

\textbf{Definition.} We have previously considered \minuncut{} which is essentially the sum of the load balancing and minimum uncut games. Similarly, \mincut{} is the sum of the load balancing and minimum cut games. We now consider a generalization which we call \mincorrelation{} which is essentially a linear combination of the three aforementioned games. Formally, it is the same as our previous models except the graph $G = (N, E^+, E^-)$ now has two disjoint types of edges: friendship edges, denoted by $E^+$, and conflict edges, denoted by $E^-$. Jobs incur additional costs for sharing conflict edges with jobs that are assigned to the same machine and for sharing friendship edges with jobs that are assigned to different machines. Furthermore, the cost incurred by load balancing, conflict edges, and friendship edges are now determined by parameters $\alpha > 0, \beta \geq 0, \gamma \geq 0$, respectively. The cost function for player $i$ with $s_{i} = k$ is defined as follows:

\begin{equation*}
c_{i}(\vec{s}) = \alpha \cdot x_{k}(\vec{s})+ \beta \cdot e^-\left(\{i \}, X_{k}(\vec{s})\right) + \gamma \cdot \sum_{l \neq k} e^+(\{i\}, X_{l}(\vec{s})).
\end{equation*}

As usual, the objective function is defined as the sum of all player costs, i.e., $c(\vec{s}) = \sum_{i \in N} c_{i} (\vec{s})$. \mincorrelation{} is also an exact potential game with $\phi(\vec{s}) = \frac{1}{2}c(\vec{s})$, which guarantees the existence of pure Nash equilibria and that the price of stability is 1.

We will now state the main result of this section, but prove it later. Using semi-smoothness, we are able to obtain tight bounds for the case in which $\alpha \geq \gamma$ and obtain 
near tight bounds when $\alpha < \gamma$. 

\begin{thm}
\label{Min Cor SS}
\mincorrelation{} is $\left(1 + \frac{m-1}{n} + \frac{\beta}{\alpha} \frac{m-1}{m} + \frac{\gamma}{\alpha}\left( \frac{m-1}{m} - \frac{m-1}{n} \right), 0 \right)$-semi-smooth when $\alpha \geq \gamma$, and the resulting price of total anarchy bound is tight. When $\alpha < \gamma$, \mincorrelation{} is $\left(1 + \frac{\beta}{\alpha} \frac{m-1}{m} + \frac{\gamma}{\alpha} \frac{m-1}{m}, 0 \right)$-semi-smooth. 
\end{thm}

We are primarily interested in the case in which $n >> m$. We observe the total price of anarchy of \mincorrelation{} asymptotically approaches $1 + \frac{\beta}{\alpha}\frac{m-1}{m} + \frac{\gamma}{\alpha}\frac{m-1}{m}$ as $n \to \infty$. 

We observe that Theorem~\ref{Min Cor SS} is a generalization of Theorems~\ref{Min Uncut SS} and \ref{Min Cut SS} and previously known results about the games that we combined to form \mincorrelation{}.  Recall that the load balancing game with identical machines has price of total anarchy $1 + \frac{m-1}{n}$, while the minimum uncut and minimum cut games have unbounded price of anarchy. Theorem~\ref{Min Cor SS} provides many insights concerning how the price of total anarchy of \mincorrelation{} is affected by combining these particular games. We observe that $\alpha$ acts as a stabilizer; increasing it results in a linear decrease in the price of total anarchy because it allows the costs incurred by the load balancing game to dominate the costs incurred by the cut games. On the other hand, increasing $\beta$ or $\gamma$ results in a linear increase in the price of total anarchy, which becomes unbounded as $\beta, \gamma \to \infty$. Essentially, the load balancing game always contributes $\frac{m-1}{n}$ to the price of total anarchy, while the cut games approximately contribute the ratio of their respective parameters and $\alpha$.


To prove Theorem~\ref{Min Cor SS}, we will now provide lower bounds for all solutions. They are straightforward generalizations of lower bounds found in Lemma~\ref{Min Uncut OPT LB 1}, Lemma~\ref{Min Uncut OPT LB 2}, and the proof of  Theorem~\ref{Min Cut SS} that follow using similar arguments, so we state them without proof.

\begin{lem} 
\label{Min Cor Lower Bounds}
For any solution $\vec{s}$, we have the following lower bounds:
\begin{align*}
c(\vec{s}) &\geq \alpha \frac{n^2}{m}   \\ 
c(\vec{s}) &\geq \frac{(\alpha - \beta(m-1))n^2}{m} + 2\beta |E ^-| \\
c(\vec{s}) &\geq (\alpha - \gamma) \frac{n^2}{m} + 2 \gamma | E ^+|  + \gamma n, \text{$\quad$ if $\alpha \geq \gamma$ } \\
c(\vec{s}) &\geq 2 \alpha | E ^+| + \gamma n, \text{$\quad$ if $\alpha \leq \gamma$.}
\end{align*}
\end{lem}

We can now prove our main result of this section by using the previous lemma to lower bound the optimal solution. 

\begin{proof}[Proof of Theorem~\ref{Min Cor SS}] For every player $i$, let $\sigma_{i}$ denote the mixed strategy in which player $i$ selects machine $k$ with probability $\frac{1}{m}$ for all $k \in M$. Using similar arguments from Theorem~\ref{Min Uncut SS} and Theorem~\ref{Min Cut SS}, we can derive that
\begin{equation*}
\mathbf{E}\left[\sum_{i \in N} c_i(\sigma_{i}, s_{-i}) \right] = \alpha\frac{n^2}{m} \left(1 + \frac{m-1}{n}\right) + 2 \beta \frac{1}{m} |E^-| + 2 \gamma \frac{m-1}{m} |E^+|.
\end{equation*}

First, we consider the case in which $|E^+| \leq \frac{n^2}{2m} - \frac{n}{2}$ and  $|E^-| \leq \frac{(m-1)n^2}{2m}$. Using these bounds, we observe that
\begin{align*}
\mathbf{E}\left[\sum_{i \in N} c_i(\sigma_{i}, s_{-i}) \right] 
&= \alpha\frac{n^2}{m} \left(1 + \frac{m-1}{n}\right) + 2 \beta \frac{1}{m} |E^-| + 2 \gamma \frac{m-1}{m} |E^+| \\
&\leq \alpha \frac{n^2}{m} \left(1 + \frac{m-1}{n}\right) +  2 \beta \frac{1}{m}\frac{(m-1)n^2}{2m} + 2 \gamma \frac{m-1}{m}\left(\frac{n^2}{2m} - \frac{n}{2} \right) \\
&= \alpha \frac{n^2}{m} \left(1 + \frac{m-1}{n} + \frac{\beta}{\alpha} \frac{m-1}{m} + \frac{\gamma}{\alpha}\left( \frac{m-1}{m} - \frac{m-1}{n} \right) \right),
\end{align*}
which completes our proof after applying the first inequality from Lemma~\ref{Min Cor Lower Bounds}.

Next, we consider the case in which $|E^+| \leq \frac{n^2}{2m} - \frac{n}{2}$ and $|E^-| \geq \frac{(m-1)n^2}{2m}$. Using these bounds and the second inequality from Lemma~\ref{Min Cor Lower Bounds}, we have that
\begin{align*}
\mathbf{E}\left[\sum_{i \in N} c_i(\sigma_{i}, s_{-i}) \right] - c(\vec{s}^*)
&\leq \alpha\frac{n^2}{m} \left(1 + \frac{m-1}{n}\right) + 2 \beta \frac{1}{m} |E^-| + 2 \gamma \frac{m-1}{m} |E^+| \\
&\quad -  \frac{(\alpha - \beta(m-1))n^2}{m} - 2\beta |E ^-| \\
&= \alpha \frac{n^2}{m}\left(\frac{m-1}{n} + \frac{\beta(m-1)}{\alpha} \right) -  2 \beta \frac{m-1}{m}|E ^-| + 2 \gamma \frac{m-1}{m}|E^+| \\
&\leq \alpha \frac{n^2}{m}\left(\frac{m-1}{n} + \frac{\beta(m-1)}{\alpha} \right) -  2 \beta \frac{m-1}{m}\frac{(m-1)n^2}{2m} \\
&\quad + 2 \gamma \frac{m-1}{m} \left(\frac{n^2}{2m} - \frac{n}{2} \right) \\
&= \alpha \frac{n^2}{m} \left(\frac{m-1}{n} + \frac{\beta}{\alpha} \frac{m-1}{m} + \frac{\gamma}{\alpha}\left( \frac{m-1}{m} - \frac{m-1}{n} \right) \right).
\end{align*}

We will now assume that $\alpha \geq \gamma$, $|E^+| \geq \frac{n^2}{2m} - \frac{n}{2}$, and $|E^-| \leq \frac{(m-1)n^2}{2m}$.  We can use these and the first and third inequalities,
\begin{align*}
\mathbf{E}\left[\sum_{i \in N} c_i(\sigma_{i}, s_{-i}) \right] - c(\vec{s}^*)
&\leq \alpha\frac{n^2}{m} \left(1 + \frac{m-1}{n}\right) + 2 \beta \frac{1}{m} |E^-| + 2 \gamma \frac{m-1}{m} |E^+| \\
&\quad - (\alpha - \gamma) \frac{n^2}{m} - 2 \gamma | E ^+| - \gamma n \\
&= \alpha \frac{n^2}{m} \left(\frac{m-1}{n} + \frac{\gamma}{\alpha} \right) + 2\beta \frac{1}{m} | E^-| - 2 \gamma \frac{1}{m} |E^+| - \gamma n \\
&\leq \alpha \frac{n^2}{m} \left(\frac{m-1}{n} + \frac{\gamma}{\alpha} \right) + 2\beta \frac{1}{m} \frac{(m-1)n^2}{2m} \\
&\quad  - 2 \gamma \frac{1}{m}\left(\frac{n^2}{2m} - \frac{n}{2} \right)  - \gamma n \\
&= \alpha \frac{n^2}{m} \left(\frac{m-1}{n} + \frac{\beta}{\alpha} \frac{m-1}{m} + \frac{\gamma}{\alpha}\left( \frac{m-1}{m} - \frac{m-1}{n} \right) \right).
\end{align*}
This completes our proof that \mincorrelation{} is $\left(1 + \frac{m-1}{n}\right. +\frac{\beta}{\alpha} \frac{m-1}{m}$ \\ $\left.+ \frac{\gamma}{\alpha}\left( \frac{m-1}{m} - \frac{m-1}{n} \right), 0 \right)$-semi-smooth when $\alpha \geq \gamma$.

Our lower bound example for which the price of total anarchy bound is tight is a combination of the lower bound examples from Claim~\ref{Min Uncut LB Examples} and Theorem~\ref{Min Cut SS}. That is, suppose there are $n = m^2$ players that are divided into $m$ disjoint sets of $m$ players each. If two players are in the same set, they have a friendship edge between them. If two players are in different sets, they have a conflict edge between them. The optimal solution is assigning each set to its own machine. Thus, $c(\vec{s}) = \alpha \frac{n^2}{m}$. Consider the mixed strategy $\sigma$ in which each player selects machine $k$ with probability $\frac{1}{m}$ for all $k \in M$. Then for each player $i \in N$, $\mathbf{E}_{\vec{s} \sim \sigma}\left[c_{i}(\vec{s}) \right] = \alpha(1 + \frac{n-1}{m}) + \beta \frac{n-m}{m} + \gamma \frac{(m-1)^2}{m}$. Clearly no player has incentive to switch, which implies that $\sigma$ is a mixed Nash equilibrium. It can be verified the price of total anarchy bounds are tight for this particular example.

Finally, we need only consider the case in which $|E^+| \geq \frac{n^2}{2m} - \frac{n}{2}$ and $|E^-| \leq \frac{(m-1)n^2}{2m}$ when $\alpha < \gamma$ to complete proof of our theorem.
\begin{align*}
\mathbf{E}\left[\sum_{i \in N} c_i(\sigma_{i}, s_{-i}) \right] 
&= \alpha\frac{n^2}{m} \left(1 + \frac{m-1}{n}\right) + 2 \beta \frac{1}{m} |E^-| + 2 \gamma \frac{m-1}{m} |E^+| \\
&\leq \alpha\frac{n^2}{m} \left(1 + \frac{m-1}{n} + \frac{\beta}{\alpha} \frac{m-1}{m}\right) + 2\gamma \frac{m-1}{m} \left(c(\vec{s}^*) - \gamma n \right) \\
&= \alpha \frac{n^2}{m} \left(1 + \frac{m-1}{n} + \frac{\beta}{\alpha}\frac{m-1}{m} - \frac{\gamma}{\alpha} \frac{m-1}{n} \right) + \frac{\gamma}{\alpha} \frac{m-1}{m} c(\vec{s}^*) \\
&\leq \alpha \frac{n^2}{m} \left(1 +  \frac{\beta}{\alpha}\frac{m-1}{m} \right) + \frac{\gamma}{\alpha} \frac{m-1}{m} c(\vec{s}^*).
\end{align*}
By our first two cases, which did not rely on the relation of $\alpha, \gamma$ and this last case, we conclude that \mincorrelation{} is $\left(1 \right. +\frac{\beta}{\alpha} \frac{m-1}{m}$ $\left.+ \frac{\gamma}{\alpha} \frac{m-1}{m}, 0 \right)$-semi-smooth when $\alpha < \gamma$.
\end{proof}

We did not consider weighted edges for any of our models, but the $\beta, \gamma$ parameters with $\alpha = 1$ manage to emulate the case in which all conflict and friendship edges have uniform weights $\beta, \gamma$, respectively. Increasing the weight of these edges causes the price of total anarchy to tend towards infinity, because the cut games start to dominate. It can be shown that if we allow arbitrary weights with the maximum weights being $\beta, \gamma$ for conflict and friendship edges, respectively, then the price of total anarchy upper bound still holds, although it will not necessarily be tight.

\subsection{Strong Nash Equilibrium for \minuncut{} with $m = 2$}

In this section we will examine strong Nash equilibrium for \minuncut{} when $m = 2$. Note that games in which people partition themselves into two groups while trying to avoid being with others they dislike already present a nontrivial and problem \cite{GM09,BCK10}. For our first result, we will show that \minuncut{} always admits a strong Nash equilibrium when $m = 2$.

\begin{thm}
The optimal solution for \minuncut{} is a strong Nash equilibrium  when $m = 2$.
\end{thm}
\begin{proof}
Suppose, by way of contradiction, that the optimal solution $\vec{s}^{*}$ is not a strong Nash equilibrium. Then there is a deviating coalition in which every player in the coalition can decrease their cost. Call the solution that results from this deviation $\vec{s}'$. Let $X_{k, l} = \lbrace i : s_{i}^{*} = k, s_{i}' = l \rbrace$, $x_{k,l} = |X_{k,l}|$. Then the deviating coalition is the set $X_{1,2} \cup X_{2,1}$ and the set of players who do not deviate is $X_{1,1} \cup X_{2,2}$.

For each player $i \in X_{1,2}$, their cost decreases by $x_{1,1} + e\left(i, X_{1,1} \right)$ and increases by $x_{2,2} + e\left(i, X_{2,2}\right)$. Since $i$ is in the deviating coalition, it must be the case that $x_{1,1} + e\left(i, X_{1,1} \right) > x_{2,2} + e\left(i, X_{2,2}\right)$. Summing this up over all players $i \in X_{1,2}$ yields
\begin{equation*}
x_{1,1}x_{1,2} + e(X_{1,1}, X_{1,2}) > x_{2,2}x_{1,2} + e(X_{2,2}, X_{1,2}).
\end{equation*} Similarly, we can derive an inequality for each player $i \in X_{2,1}$. Summing this inequality over all players $i \in X_{2,1}$ results in
\begin{equation*}
x_{2,2}x_{2,1} + e(X_{2,2}, X_{2,1}) > x_{1,1}x_{2,1} + e(X_{1,1}, X_{2,1}).
\end{equation*}
Next, we consider a player $i \in X_{1,1}$. Their cost decreases by $x_{1,2} + e\left(i, X_{1,2} \right)$ and increases by $x_{2,1} + e(i, X_{2,1})$. Then $\sum_{i \in X_{1,1}} \left(c_{i}(\vec{s}') - c_{i}(\vec{s}^{*}) \right)= x_{1,1} x_{2,1} + e(X_{1,1}, X_{2,1}) - x_{1,1} x_{1,2} - e(X_{1,1}, X_{1,2})$.

Similarly, $\sum_{i \in X_{2,2}} \left(c_{i}(\vec{s}') - c_{i}(\vec{s}^{*}) \right)= x_{2,2} x_{1,2} + e(X_{2,2}, X_{1,2}) - x_{2,2} x_{2,1} - e(X_{2,2}, X_{2,1})$.

We can now calculate $c(\vec{s}') - c(\vec{s}^{*})$.
\begin{align*}
c(\vec{s}') - c(\vec{s}^{*}) &= \sum_{i \in N} \left(c_{i}(\vec{s}') - c_{i}(\vec{s}^{*}) \right) \\
&= 2 \left(x_{2,2} x_{1,2} + e(X_{2,2}, X_{1,2}) + x_{1,1} x_{2,1} + e(X_{1,1}, X_{2,1}) \right.  \\
&\quad \left. - x_{1,1} x_{1,2} - e(X_{1,1}, X_{1,2}) - x_{2,2} x_{2,1} - e(X_{2,2}, X_{2,1})  \right) \\
&< 0
\end{align*}
by our inequalities, which is a contradiction.
\end{proof}

It immediately follows from this theorem that the strong price of stability is 1. We suspect that strong Nash equilibrium need not exist for $m \geq 3$, but we've yet to find any such examples. The existence of strong Nash equilibrium for cut games in general is unresolved, as seen in \cite{GM09, GM10}.

Next, we will provide an upper bound for the strong price of anarchy with the same technique used to bound the strong price of anarchy for the max $k$-cut game in \cite{GM10}. This technique fails when $m > 2$, partially because the technique in \cite{GM10} is for max $k$-cut as opposed to min $k$-uncut, which is the type of cut problem that we are considering here.

\begin{thm}
The strong price of anarchy is at most $\frac{4}{3} + \frac{2}{3n}$ \minuncut{} with $m = 2, n \geq 4$. The strong price of anarchy is 1 for $m = 2, n \leq 3$.
\end{thm}
\begin{proof}
Let $\vec{s}^{*}$ denote an optimal solution and let $\vec{s}'$ denote a permutation of $\vec{s}^{*}$, i.e., $X_{1}(\vec{s}^{*}) = X_{2}(\vec{s}')$ and $X_{2}(\vec{s}^{*}) = X_{1}(\vec{s}')$. Let $\vec{s}$ be a strong Nash equilibrium. Let $X_{k, l} = \lbrace i : s_{i} = k, s_{i}^{*} = l \rbrace$, $x_{k,l} = |X_{k,l}|$. Then
\begin{align*}
c(\vec{s}) &= \left(x_{1,1} + x_{1,2}\right)^{2} + \left(x_{2,1} + x_{2,2}\right)^{2} + 2\left( \sum_{k=1}^{2} \sum_{l=1}^{2} e(X_{k,l}) + e(X_{1,1}, X_{1,2}) + e(X_{2, 2}, X_{2,1}) \right) \\
c(\vec{s}^{*}) = c(\vec{s}') &= \left(x_{1,1} + x_{2,1}\right)^{2} + \left(x_{1,2} + x_{2,2}\right)^{2} + 2\left( \sum_{k=1}^{2} \sum_{l=1}^{2} e(X_{k,l}) + e(X_{1,1}, X_{2,1}) + e(X_{2, 2}, X_{1,2}) \right).
\end{align*}
Let $X_{1,2} \cup X_{2,1}$ be a deviating coalition in $\vec{s}$. By the definition of $X_{k, l}$ and the fact there are only two strategies, the resulting state of this deviation is $\vec{s}^{*}$. However, since $\vec{s}$ is a strong Nash equilibrium, that must mean there is a player $i_{1} \in X_{1,2} \cup X_{2,1}$ that does not want to deviate with the coalition. Furthermore, there must be a player $i_{2}$ that does not want deviate with the coalition $X_{1,2} \cup X_{2,1} \backslash \{ i_{1} \}$. We can repeat this process for each player in $X_{1,2} \cup X_{2,1}$, which gives us an inequality for each player. Suppose that player $i_{t} \in X_{1, 2}$ is the $t$'th player to leave the coalition in this way. Let $Y_{1,2}^{t}$ denote the set of players in $X_{1,2}$ that left the coalition before $i_{t}$ and $Y_{2,1}^{t}$ players in $X_{2,1}$ that left the coalition before $i_{t}$. Let $y_{1,2}^{t} = |Y_{1,2}^{t}|$ and $y_{2,1}^{t} = |Y_{2,1}^{t}|$. Then
\begin{equation*}
c_{i_t}(\vec{s}) \leq x_{2, 2} + x_{1, 2} - y_{1,2}^{t} + y_{2,1}^{t} + e\left(\{i_t\}, X_{2,2}\right) + e\left(\{i_t\}, X_{1,2} \backslash Y_{1,2}^{t}\right) + e\left(\{i_t\}, Y_{2,1}^{t}\right).
\end{equation*}
Similarly, if player $i_{t}$ is in $X_{2,1}$, then
\begin{equation*}
c_{i_t}(\vec{s}) \leq x_{1, 1} + x_{2, 1} - y_{2,1}^{t} + y_{1,2}^{t} + e\left(\{i_t\}, X_{1,1}\right) + e\left(\{i_t\}, X_{2,1} \backslash Y_{2,1}^{t}\right) + e\left(\{i_t\}, Y_{1,2}^{t}\right).
\end{equation*}
Our goal is to find an upper bound for $\sum_{i \in X_{1,2} \cup X_{2,1}} c_{i}(\vec{s})$. We will do this by summing our previous inequalities over all players in $X_{1,2} \cup X_{2,1}$. We claim that doing this gives us
\begin{align}
\sum_{i \in X_{1,2} \cup X_{2,1}} c_{i}(\vec{s}) &\leq x_{2,2}x_{1,2} + x_{1,1}x_{2,1} + \frac{x_{1,2}^{2}+x_{1,2}}{2} + \frac{x_{2,1}^{2}+x_{2,1}}{2} + x_{1,2}x_{2,1} \notag  \\
&\quad + e(X_{2,2}, X_{1,2}) + e(X_{1,1}, X_{2,1}) + e(X_{1,2}) + e(X_{2,1}) + e(X_{1,2}, X_{2,1}) \label{SE UB 1}.
\end{align}

To prove Inequality (\ref{SE UB 1}), we will consider each term separately. Clearly, $\sum_{i \in X_{1,2}} x_{2,2} = x_{2,2}x_{1,2}$ and $\sum_{i \in X_{2,1}} $ $ x_{1,1} = x_{1,1}x_{2,1}$. Similarly, we observe that $\sum_{i\in X_{1,2}} $ $ e(i, X_{2,2}) = e(X_{1,2}, X_{2,2})$ and $\sum_{i\in X_{2,1}} e(i, X_{1,1}) = e(X_{2,1}, X_{1,1})$.

Next, we consider $\sum_{i_{t} \in X_{1,2}} (x_{1,2} - y_{1,2}^{t})$. Consider the first player $i_{t}$ in $X_{1,2}$ to leave the coalition. Then $y_{1,2}^{t} = 0$, since no player in $X_{1,2}$ has left the coalition before $i_{t}$. For the second player $i_{t'}$ in $X_{1,2}$ to leave the coalition,  $y_{1,2}^{t'} = 1$, because player $i_{t} \in Y_{1,2}^{t'}$. If we continue to repeat this process, we observe that for the $k$th player $i_{t}$ in $X_{1,2}$ to leave the coalition, $y_{1,2}^{t} = k - 1$. Thus, it follows that
\begin{equation*}
\sum_{i_{t} \in X_{1,2}} \left(x_{1,2} -  y_{1,2}^{t} \right) = \sum_{k=0}^{x_{1,2}-1} \left(x_{1,2} - k\right)= \frac{x_{1,2}^{2}+x_{1,2}}{2}.
\end{equation*}
Now we consider $\sum_{i_{t} \in X_{1, 2}} e(\{i_{t} \}, X_{1,2}\backslash Y_{1,2}^{t})$. Consider an edge $\left(i_{t}, i_{t'}\right)$ where $i_{t}, i_{t'} \in X_{1, 2}$, $t < t'$. Since $i_{t'} \notin Y_{1,2}^{t}$,  $(i_{t}, i_{t'}) \in E(i_{t}, X_{1,2}\backslash Y_{1,2}^{t})$ which means that the edge is counted for player $i_{t}$ in the sum above. However, since $i_{t} \in Y_{1,2}^{t'}$, $(i_{t}, i_{t'}) \notin E(i_{t'}, X_{1,2}\backslash Y_{1,2}^{t'})$ which means that the edge is not counted for player $i_{t'}$ in this sum. Thus, we conclude that each edge in $E(X_{1,2})$ is counted exactly once. That is,
\begin{equation*}
\sum_{i_{t} \in X_{1,2}} e\left(\{ i_t \}, X_{1,2} \backslash Y_{1,2}^{t}\right)  =  e\left(X_{1,2}\right).
\end{equation*}
Similar analysis holds for players in $X_{2,1}$. That is, $\sum_{i_{t} \in X_{2,1}} \left(x_{2,1} -  y_{2,1}^{t} \right) = \frac{x_{2,1}^2 + x_{2,1}}{2}$ and $\sum_{i_{t} \in X_{2,1}}$ $e\left(\{ i_t \}, X_{2,1} \backslash Y_{2,1}^{t}\right)$  $=  e\left(X_{2,1}\right)$.

We will now calculate $\sum_{i_{t} \in X_{1,2}} y_{2,1}^{t} + $ $\sum_{i_{t} \in X_{2,1}} y_{1,2}^{t}$. Consider a pair of players $i_{t} \in X_{1,2}$ and $i_{t'} \in X_{2,1}$. Suppose, without loss of generality, that $i_{t}$ leaves the coalition first. Since $i_{t'} \notin Y_{2,1}^{t}$ and $i_{t} \in Y_{1,2}^{t'}$, this pair contributes 1 to our sum. Since this holds for every pair of players and there are $x_{1,2}x_{2,1}$ such pairs, we conclude that $\sum_{i_{t} \in X_{1,2}} y_{2,1}^{t} + $ $\sum_{i_{t} \in X_{2,1}} y_{1,2}^{t} =  x_{1,2}x_{2,1}$. Similarly, $\sum_{i_{t} \in X_{1,2}} e(\{i_{t}\}, Y_{2,1}^{t}) + $ $\sum_{i_{t} \in X_{2,1}} e(\{i_{t}\}, Y_{1,2}^{t}) = e(X_{1,2}, X_{2,1})$.

Combining these results gives us (\ref{SE UB 1}).

We will now consider another deviating coalition, $X_{1,1} \cup X_{2,2}$, which results in the state $\vec{s}'$. The exact same analysis holds for this coalition, giving us
\begin{align}
\sum_{i \in X_{1,1} \cup X_{2,2}} c_{i}(\vec{s}) &\leq x_{2,2}x_{1,2} + x_{1,1}x_{2,1} + \frac{x_{1,1}^{2}+x_{1,1}}{2} + \frac{x_{2,2}^{2}+x_{2,2}}{2} + x_{1,1}x_{2,2} \notag  \\
&\quad + e(X_{2,2}, X_{1,2}) + e(X_{1,1}, X_{2,1}) + e(X_{1,1}) + e(X_{2,2}) + e(X_{1,1}, X_{2,2}) \label{SE UB 2}.
\end{align}
We can combine (\ref{SE UB 1}) and (\ref{SE UB 2}) to give us an upper bound on $c(\vec{s})$. Using this inequality and the value of $c(\vec{s}^{*})$, we can derive
\begin{align*}
c(\vec{s}) &\leq \sum_{k=1}^{2}\sum_{l=1}^{2} \left(\frac{x_{k, l}^{2} + x_{k,l}}{2} + e\left(X_{u,v} \right) \right) + 2x_{2,2}x_{1,2} + 2x_{1,1}x_{2,1} + x_{1,2}x_{2,1} \\
&\quad + x_{1,1}x_{2,2} + 2e\left(X_{2,2}, X_{1,2}\right) + 2e\left(X_{1,1}, X_{2,1}\right) +e\left(X_{1,2}, X_{2,1}\right) + e\left(X_{1,1}, X_{2,2}\right) \\
&= \frac{1}{2}c(\vec{s}^{*}) + \frac{n}{2} + x_{2,2}x_{1,2} + x_{1,1}x_{2,1} + x_{1,2}x_{2,1}  + x_{1,1}x_{2,2} \\
&\quad+ e\left(X_{2,2}, X_{1,2}\right) + e\left(X_{1,1}, X_{2,1}\right) +e\left(X_{1,2}, X_{2,1}\right) + e\left(X_{1,1}, X_{2,2}\right) \\
&= \frac{1}{2}c(\vec{s}^{*}) + \frac{n}{2} + x_{2,2}x_{1,2} + x_{1,1}x_{2,1} + x_{1,2}x_{2,1}  + x_{1,1}x_{2,2} \\
&\quad+ |E| - e\left(X_{1,1}, X_{1,2} \right) - e\left(X_{2,1}, X_{2,2}\right) - \sum_{k=1}^{2}\sum_{l=1}^{2} e\left(X_{k, l} \right)
\end{align*}
Adding $\frac{1}{2}c(\vec{s})$ to both sides, and applying Lemmas \ref{Min Uncut OPT LB 1} and \ref{Min Uncut OPT LB 2} allows us to derive
\begin{align*}
\frac{3}{2}c(\vec{s}) &\leq \frac{1}{2}c(\vec{s}^{*}) + \frac{n}{2} +\frac{1}{2}\sum_{k=1}^{2}\sum_{l=1}^{2} x_{k, l}^{2} + x_{1,1}x_{1,2}+ x_{2,1}x_{2,2}   \\
&\quad + x_{2,2}x_{1,2} + x_{1,1}x_{2,1} + x_{1,2}x_{2,1}  + x_{1,1}x_{2,2} + |E| \\
&= \frac{1}{2}c(\vec{s}^{*}) + \frac{n}{2} + \frac{1}{2}(x_{1,1} + x_{1,2} + x_{2,1} + x_{2,2})^{2} + |E| \\
&= \frac{1}{2}c(\vec{s}^{*}) + \frac{n}{2} + \frac{n^2}{2} + |E| \\
&\leq \left(2 + \frac{1}{n} \right)c(\vec{s}^{*})
\end{align*}
which completes the proof.
\end{proof}

Unfortunately, this bound is not tight; the best lower bound we have found for the strong price of anarchy is $\frac{5}{4}$, using the same example as the one used in \cite{GM09}.

\begin{claim}
The strong price of anarchy of \minuncut{} with $m=2$ is at least $\frac{5}{4}$.
\end{claim}
\begin{proof}
Let $G$ be a path on four nodes. The solution in which two jobs are assigned to each machine such that no job incurs edge cost is the optimal solution with $c(\vec{s}^{*}) = 8$.

Consider the solution $\vec{s}$ formed by assigning the two degree one jobs to one machine and the other two jobs to the other machine. Then the degree one jobs each have cost 2 and the other two jobs have cost 3, which means $c(\vec{s}) = 10$. Since neither of the jobs with cost 2 will join a deviating coalition, this solution is a strong Nash equilibrium.
\end{proof}

Finally, we give the corresponding results for \mincut{}.

\begin{thm}
The optimal solution is a strong Nash equilibrium for \mincut{} with $m = 2$, and the strong price of anarchy is at most $\frac{4}{3}$.
\end{thm}

\section{\maxcut}



\textbf{Definition.} We are given a set of players $N = \{1, 2, \dots, n \}$, a set of machines $M = \{1, 2, \dots, m \}$ with positive values $P = \{p_1,p_2, \dots , p_m\}$, and $G = (N, E)$, an undirected graph on the player set.  Each player must assign a job to a single machine to be processed. 


We will use the same notation as we did in \minuncut{}. Each machine $k$ will distribute the $p_k$ utility associated with it equally among the jobs assigned to it. Each job will receive 1 utility for each edge $e$ that it shares with another job that is assigned to another machine. We define the utility of job $i$ with $s_{i} = k$ as  $u_i(\vec{s}) = \frac{p_{k}}{x_{k}(\vec{s})} + \sum_{l \neq k} e(\{i\}, X_{l}(\vec{s}))$. Finally, we define the utility of a solution $\vec{s}$, denoted $u(\vec{s})$, as the sum of the utility of all jobs. That is, $u(\vec{s}) = \sum_{i \in N} u_{i}(\vec{s})$.




\subsection{Tight Bounds for Quality of CCE}

We will now provide analysis of the \maxcut{} game. Observe that \maxcut{} is an exact potential game with $\Phi(\vec{s}) = \sum_{k=1}^{m} \sum_{l=1}^{x_{k}(\vec{s})} \frac{p_{k}}{l} + \sum_{k=1}^{m-1} \sum_{l=k+1}^{m} e\left(X_{k}(\vec{s}), X_{l}(\vec{s})\right)$, which guarantees the existence of pure Nash equilibria. We will now prove the main result of this section: \maxcut{} is \sse. This, along with our lower bound example, will give tight price of anarchy bounds for many equilibrium concepts.

\begin{thm}
\label{Max Cut SS}
\maxcut{} is $\left(\frac{m-1}{m}, \frac{m-2}{m} \right)$\emph{-\sse}.
\end{thm}
\begin{proof}
Let $\vec{s}$ be an outcome, and let $d_i$ denote the degree of player $i$ in $G$. We will assume that $n \geq m$, but this argument works similarly for $n < m$. Suppose, without loss of generality, that only the first $t$ machines have players assigned to them in $\vec{s}$. Then $u(\vec{s}) \geq \sum_{k = 1}^{t} p_{k}$. For each player $i$ with $s_{i} = k$, $u_{i}(\vec{s}) = p_{k} \cdot x_{k}(\vec{s})^{-1} + \sum_{l \neq k} e(\{ i\}, X_l(\vec{s}))$. If player $i$ deviates to a different machine $l$, then their utility will be the sum of $p_{l}$ divided by the number of players already on machine $l$ plus 1 for itself and the number of edges from $i$ to players not on $l$. That is, for each $s_{i}' = l \neq s_{i}$, $u_{i}(s_{i}', s_{-i}) = p_{l} \cdot (x_{l}(\vec{s}) + 1)^{-1} + \sum_{k \neq l} e(\{i\}, X_{k}(\vec{s}))$. For each player $i$, let $\sigma_i$ denote the mixed strategy in which every strategy $s_{i}' \in S_{i}$ is selected with probability $\frac{1}{m}$. Then using the fact that $x_{l}(\vec{s})$ = 0 for $l > t$ allows us to derive that for every player $i$ with $s_{i} = k$,
\begin{align*}
\mathbf{E}[u_{i}(\sigma_{i}, s_{-i})]
&= \frac{1}{m} \left(\sum_{l = 1: l \neq k}^{m} \frac{p_{l}}{x_{l}(\vec{s}) + 1} + \frac{p_{k}}{x_{k}(\vec{s})} + \sum_{l = 1}^{m} \sum_{l' \neq l} e(\{i\}, X_{l'}(\vec{s})) \right).
\end{align*}
It is easy to see that $\sum_{l = 1}^{m} \sum_{l' \neq l} e(\{i\}, X_{l'}(\vec{s}))=(m-1)d_{i}$, since all edges from $i$ to each machine $l$ appear $m-1$ times in the sum. Notice also that for machines $l>t$, we know that $x_l(s)=0$ by definition of $t$, and so $\frac{p_{l}}{x_{l}(\vec{s}) + 1}=p_l$. Thus,
\begin{align*}
\mathbf{E}[u_{i}(\sigma_{i}, s_{-i})]
&= \frac{1}{m} \left(\sum_{l = 1: l \neq k}^{t} \frac{p_{l}}{x_{l}(\vec{s}) + 1} + \sum_{l = t+1}^{m} p_{l} + \frac{p_{k}}{x_{k}(\vec{s})} + (m-1)d_{i} \right) \\
&\geq \frac{1}{m}\left(\frac{p_{k}}{x_{k}(\vec{s})} + \sum_{l=t+1}^{m} p_{l} + (m-1)d_{i} \right).
\end{align*}
By linearity of expectation and the fact that $\sum_{i : s_{i} = k} \frac{p_{k}}{x_{k}(\vec{s})} = p_{k}$,
\begin{align*}
\mathbf{E}\left[\sum_{i \in N} u_{i}( \sigma_{i}, s_{-i})\right]
&= \sum_{i \in N} \mathbf{E}[u_{i}(\sigma_{i}, s_{-i})]  \\
&\geq \frac{1}{m}\left(\sum_{l=1}^{t} p_{l} + n \sum_{l=t+1}^{m}p_{l} + 2(m-1)|E| \right)
\end{align*}
It is clear that for any solution $\vec{s}'$, $u(\vec{s}') \leq \sum_{l=1}^{m} p_{l} + 2|E|$. Using this, our lower bound for $u(\vec{s})$, and the fact that $n \geq m$, we can show our previous quantity is at least
\begin{align*}
\frac{1}{m}\left(\sum_{l=1}^{t} p_{l} + 2(m-1)|E| \right) + \frac{m-1}{m} \sum_{l=t+1}^{m} p_{l}
&\geq \frac{m-1}{m}\left(\sum_{l=1}^{m} p_{l} + 2|E| \right) - \frac{m-2}{m} \sum_{l=1}^{t} p_{l} \\
&\geq \frac{m-1}{m} u(\vec{s}^{*}) - \frac{m-2}{m} u(\vec{s}).
\end{align*}
The same argument works when $n < m$ because we can ignore the $m - n$ machines with the least value. Players only need to deviate to the $n$ machines with the highest value, and the upper bound we obtain for OPT is lower since at most $n$ machines can be covered.
\end{proof}

We will now give a lower bound for the price of stability, which also acts as a lower bound for the price of anarchy, to show that the upper bound we found with \sseness{} is tight.

\begin{claim}
\label{Max Cut PoA LB}
The price of stability of the \maxcut{} game is at most 2, and it is tight.
\end{claim}
\begin{proof}
Let $n = m$, $G = K_{m}$, $p_{1} = m^2 - m + \epsilon$ for any $\epsilon > 0$ and for all machines $k \neq 1$, $p_{k} = 0$. The optimal solution is to assign each job to a different machine, in which case every job receives $m-1$ utility from edges and the job assigned to machine 1 receives an additional $m^2 - m + \epsilon$ utility from the machine. Thus, $u(\vec{s}^{*}) = 2m^2 - 2m + \epsilon$.


Any player is guaranteed to receive at least $m-1 + \frac{\epsilon}{m}$ utility if they choose machine 1. However, each player can receive at most $m-1$ utility if they choose any of the other machines. We conclude that the solution $\vec{s}$ formed by assigning all jobs to machine 1 is the only pure Nash equilibrium. We observe that $u(\vec{s}) = m^2 - m + \epsilon$ which approaches $\frac{1}{2}u(\vec{s}^{*})$ as $\epsilon \to 0$.

Theorem~\ref{Max Cut SS} guarantees that the price of anarchy is at most 2, which implies the price of stability cannot exceed 2.
\end{proof}

This same example also shows that the strong price of stability is 2. The following price of anarchy results follow immediately from Theorem~\ref{Max Cut SS} and our lower bound example from Claim~\ref{Max Cut PoA LB}:

\begin{cor}
The price of total anarchy is at most 2, and it is tight. The same holds for the strong price of anarchy.
\end{cor}

This result is interesting because the price of anarchy for the market sharing game with $m$ machines is $2 - \frac{1}{m}$ and the price of anarchy of the max $k$-cut game is $\frac{k}{k-1}$. For $m > 2$, this is lower than the price of stability (and anarchy)  of \maxcut{}. Since our game is essentially a combination of these two games, this means that combining two games can lead to greater inefficiency.

Finally, we will consider convergence for the \maxcut{} game. Our convergence results for \maxcut{} are weaker than those for \minuncut{} because \maxcut{} is not a perfect game, which means we cannot use the convergence result from \cite{ACE+11}. Instead, we use a generalization of the convergence result from \cite{R09} to obtain our first result: that a good solution is reached very quickly. Then we use a modified version of the convergence result from \cite{ACE+11} that does not require the game to be perfect. The trade-off is that the approximation factor of the solutions reached is not as good. We observe that $\log (n) \cdot \Phi(\vec{s}) \geq u(\vec{s}) \geq \frac{1}{2}\Phi(\vec{s})$ for this game. Thus, by Theorems~\ref{Max Cut SS}, \ref{Utility Convergence 1}, and \ref{Utility Convergence 2}, we have the following two convergence results.

\begin{cor}
For any \maxcut{} instance and for any $\epsilon > 0$, the BR dynamic reaches a state $\vec{s}^{t}$ with $u(\vec{s}^{t}) \geq \frac{1}{2+2\epsilon} \cdot u(\vec{s}^{*})$ from any starting state $\vec{s}^{0}$ in at most $O\left(\frac{n }{\epsilon} \log u(\vec{s}^{*}) \right)$ steps.
\end{cor}

\begin{cor}
For any \maxcut{} instance and for any $\epsilon > 0$, the BR dynamic converges to a state $\vec{s}^{t}$ with $u(\vec{s}^{t}) \geq \frac{1-\epsilon}{4\log (n)} \cdot u(\vec{s}^{*})$  in at most $O\left(\frac{n}{\log (n)} \log \frac{1}{\epsilon} \right)$ steps from any initial state. Furthermore, all future states reached with best-response dynamics will satisfy this approximation factor as well.
\end{cor}

\subsection{Extensions}
\label{sec:extensions2}

\subsubsection{Weighted Edges}

All of our results in the previous section hold using the same arguments if we allow the edges to have positive weights and redefine the utility function of a job $i$ such that it receives $w_{e}$ utility for each edge $e$  that it shares with another job that is assigned to another machine. Thus, unlike in Section \ref{sec:costs}, the bounds on the efficiency of equilibrium remain small (at most 2) even in the presence of weighted edges.

\subsubsection{\maxuncut}

%
%

\textbf{Definition.} Just as we considered a utility version of the \minuncut{} game, we will now consider a utility version of the \mincut{} game. The \maxuncut{} game is the same as the \maxcut{} game except that jobs receive additional utility for sharing edges with jobs on the same machine. That is, for every job $i$ with $s_{i} = k$, $u_{i}(\vec{s}) = \frac{p_{k}}{x_{k}(\vec{s})} + e(\{i\}, X_{k}(\vec{s}))$.

We will prove similar results for the \maxuncut{} using the same arguments as we did for the \maxcut{}. Furthermore, all of these results hold if we allow arbitrary positive edge weights. \maxuncut{} is an exact potential game with $\Phi(\vec{s}) = \sum_{k=1}^{m} \sum_{l=1}^{x_{k}(\vec{s})} \frac{p_{k}}{l} + \sum_{k=1}^{m}  e_{k}(\vec{s})$.

\begin{thm}
\label{Max Uncut SSE}
The \maxuncut{} is $\left(\frac{1}{m}, 0 \right)$\emph{-\sse}. The price of total anarchy is at most $m$, and it is tight. The price of stability is at least $2 - \frac{1}{m}$.
\end{thm}

\begin{cor}
For any \maxuncut{} instance and for any $\epsilon > 0$, the BR dynamic reaches a state $\vec{s}^{t}$ with $u(\vec{s}^{t}) \geq \frac{1}{m(1+\epsilon)} \cdot u(\vec{s}^{*})$ from any starting state $\vec{s}^{0}$ in at most $O\left(\frac{n }{\epsilon} \log u(\vec{s}^{*}) \right)$ steps.
\end{cor}

\begin{cor}
For any \maxuncut{} instance and for any $\epsilon > 0$, the BR dynamic reaches to a state $\vec{s}^{t}$ with $u(\vec{s}^{t}) \geq \frac{(1-\epsilon)}{2m\log (n)} \cdot u(\vec{s}^{*})$  in at most $O\left(\frac{n}{\log (n)} \log \frac{1}{\epsilon} \right)$ steps from any initial state. Furthermore, all future states reached with best-response dynamics will satisfy this approximation factor as well.
\end{cor}

Unlike in \maxcut, where existence of strong Nash equilibrium is an open question even for $m=2$, we can show that strong Nash equilibrium does not necessarily exist for this game. The strong price of anarchy is known to be at most 3 \cite{BSTV13}, and our price of stability example gives a lower bound of 2.

\begin{thm}
Strong Nash equilibrium need not exist for the \maxuncut{} game.
\end{thm}
\begin{proof}
Consider an instance in which $G$ is a graph with 4 nodes and the edges $(1, 2)$ and $(3, 4)$. There are two machines with values $2 + \epsilon$ and $4 + 3\epsilon$ for some $\epsilon > 0$. It is easily verified that no strong Nash equilibrium exists for this instance.
\end{proof}

\section{Conclusion and Future Work}
\label{Conclusion and Future Work}

In this paper, we introduce several new job-assignment games, such as Balancing with Conflicts, Balancing with Friendship,  and Sharing with Conflicts, which allow us to model positive and negative interactions between certain pairs of jobs.  For both games, we provide tight or nearly-tight bounds on the price of anarchy for pure Nash, mixed Nash, correlated, coarse correlated, and strong Nash equilibria.  We also bound the time needed for  convergence to near-optimal outcomes.  For most of these results, we make use of the notion of semi-smoothness \cite{CKK+12} and observe that standard smoothness \cite{R09} is not sufficient for achieving tight bounds.

There are a number of natural avenues for further study.  Some of the more direct extensions include generalizing our existing results to broader classes of latency functions; in the case of balancing with conflicts, proving tight bounds on the price of anarchy seems difficult even for related machines, and may require new techniques.

Another promising direction is to improve bounds on the quality of equilibria when the conflict graph's structure is constrained.  For example, consider our motivating scenario in which jobs conflict over the use of certain computational resources.  If the number of distinct resource types is small, the set of possible conflict graphs is limited and thus better bounds may apply.  More generally, it would be interesting to find any structural parametrization that increases equilibrium quality.

Perhaps most interesting would be a broader exploration of the process of combining games.  All of the games studied in this paper can be viewed as a ``sum'' of simpler games.  Given any pair of games $\mathcal{G}_1$ and $\mathcal{G}_2$ in which the players and their strategy sets are the same, we can define a new game $\mathcal{G}_1 \oplus \mathcal{G}_2$ with the same players, the same strategies, and utilities given by $u_i(\vec{s}) = u^1_i(\vec{s})+u^2_i(\vec{s})$, where $u^j_i$ is $i$'s utility function in $\mathcal{G}_j$.
In particular, the games we introduced are sums of load balancing games, market sharing games, and cut games.   To what extent do these new games inherit properties of their component games?  Some relationships are straightforward -- if $\mathcal{G}_1$ and $\mathcal{G}_2$ are potential games, then clearly so too is $\mathcal{G}_1 \oplus \mathcal{G}_2$.  Similarly, one can bound the smoothness of $\mathcal{G}_1 \oplus \mathcal{G}_2$ in terms of that of $\mathcal{G}_1$ and $\mathcal{G}_2$ (although the naive bounds do not appear to be tight).  Under what conditions and to what extent are various game characteristics (such as smoothness) preserved under this game operator?  Are there natural classes of games that are well behaved under this operation?  The same questions can be applied to other game operators as well, such as product, min, and max.  The development of tools for analyzing such operators is a promising direction for future research.

\section*{Acknowledgements} We thank Carlos Varela for illuminating discussions on the subject of job assignment.

\end{document}